\documentclass[journal]{IEEEtran}

\ifCLASSINFOpdf
\else
\fi

\usepackage{cite}
\usepackage{amsmath,amssymb,amsfonts}
\usepackage[linesnumbered, ruled, vlined]{algorithm2e}
\usepackage{algpseudocode}
\usepackage{graphicx}
\usepackage{textcomp}
\usepackage{lipsum}
\usepackage{amsthm}
\usepackage{siunitx}
\usepackage{amsmath}
\usepackage{subfigure}
\usepackage{multicol}
\usepackage{multirow}
\usepackage{mathtools}
\usepackage{color}
\usepackage{bm}
\usepackage{indentfirst}
\usepackage{nomencl}
\usepackage{booktabs}
\usepackage{cite}
\usepackage{mathtools}
\usepackage{accents}
\usepackage{tikz}
\usepackage[numbers,sort&compress]{natbib}

\usepackage[hyphens]{url}
\usepackage[hidelinks]{hyperref}
\usepackage{makecell}

\newtheorem{theorem}{Theorem}

\newtheorem{remark}{Remark}

\SetKwInOut{Input}{Input}
\SetKwInOut{Output}{Output}

\newcommand{\LL}{\mathcal{L}}
\newcommand{\B}{\mathbf{B}}
\newcommand{\R}{\mathbb{R}}
\newcommand{\X}{\mathbf{X}}
\newcommand{\x}{\mathbf{x}}

\newcommand{\PP}{\mathcal{P}}
\newcommand{\smallsum}{{\textstyle \sum}}

\newcommand{\ubar}[1]{\underline{#1}}
\renewcommand{\bar}[1]{\overline{#1}}

\begin{document}

\title{Optimal Network Charge  for Peer-to-Peer Energy Trading:  A Grid Perspective}

\author{Yu~Yang,~\IEEEmembership{Member,~IEEE,}
	Yue~Chen, ~\IEEEmembership{Member,~IEEE,} \\
	Guoqiang~Hu,~\IEEEmembership{Senior Member,~IEEE,}
	and~Costas~J.~Spanos,~\IEEEmembership{Fellow,~IEEE}
	\thanks{This  work  was  supported  by  the  Republic  of  Singapore’s  National  Research  Foundation  through  a  grant  to  the  Berkeley  Education  Alliance  for  Research  in  Singapore
		(BEARS)  for  the  Singapore-Berkeley  Building  Efficiency  and  Sustainability  in  the
		Tropics  (SinBerBEST)  Program.  BEARS  has  been  established  by  the  University  of  California,  Berkeley  as  a  center  for  intellectual  excellence  in  research  and  education  in
		Singapore.}
	\thanks{Y. Yang is with SinBerBEST, Berkeley Education 	Alliance for Research in Singapore, Singapore 138602 e-mail: (yangyu13@tsinghua.org.cn).}
	\thanks{Y. Chen is with the Department of Mechanical and Automation Engineering, the Chinese University of Hong Kong, Hong Kong SAR, China. (e-mail: yuechen@mae.cuhk.edu.hk). The work of Y. Chen was supported by CUHK research startup fund.}
	\thanks{G. Hu is with the School 	of Electrical and Electronic Engineering, Nanyang Technological University,
		Singapore, 639798 e-mail: (gqhu@ntu.edu.sg).}
	\thanks{C. J. Spanos is with the Department of Electrical Engineering and 	Computer Sciences, University of California, Berkeley, CA, 94720 USA email: (spanos@berkeley.edu).}
}


\maketitle

\begin{abstract}
Peer-to-peer (P2P) energy trading is a promising market scheme to accommodate the increasing distributed energy resources (DERs). However, how P2P  to be integrated into the existing power systems remains to be investigated. In this paper, we apply network charge as a means for the grid operator to attribute transmission loss and ensure network constraints for empowering P2P transaction. The interaction between the grid operator and the prosumers is modeled as a Stackelberg game, which yields a bi-level optimization problem. We prove that the Stackelberg game admits an \emph{equilibrium} network charge price. Besides, we propose a method to obtain the network charge price by converting the bi-level optimization into a single-level mixed-integer quadratic programming (MIQP), which can handle a reasonable scale of prosumers efficiently. Simulations on the IEEE bus systems show that the proposed optimal network charge is favorable as it can benefit both the grid operator and the prosumers for empowering the P2P market, and achieves \emph{near-optimal} social welfare. Moreover, the results show that the presence of energy storage will make the prosumers more sensitive to the network charge price changes.
\end{abstract}


\IEEEpeerreviewmaketitle

\section{Introduction}
\IEEEPARstart{D}{riven}  by  the technology advances and the pressure to advance low-carbon society,  power systems are experiencing the steady increase of  
distributed energy resources (DERs), such as home batteries, electric vehicles (EVs), roof-top solar panels, and on-site wind turbines, etc. \cite{DERs, yang2017distributed, yang2018decentralized, yang2020selling}.
As a result, the  traditional centralized energy management  is being challenged as the DERs on the  customer side are beyond the control of the power grid operator. 
In this context, peer-to-peer (P2P) energy  trading has emerged as a promising mechanism to account for the DERs  \cite{morstyn2018using}.  
P2P aims for a consumer-centric  electricity market that allows the consumers with DERs (i.e., prosumer) to trade   energy surplus or deficiency mutually  \cite{zhang2018peer, chen2021energy, chen2022towards}. 
The vision of P2P is to empower the prosumers to achieve the balance of  supply and demand  autonomously and economically by leveraging their complementary and flexible  generation and consumption. 
P2P energy  trading  is beneficial to both the power grid operator and the prosumers.   
Specifically,  P2P  can  bring  monetary value  to the prosumers  by allowing them to sell surplus local  renewable generation to their neighbors or vice verse \cite{tushar2018transforming, yang2021optimal}.  
P2P also favors the power grid operation  in term of   reducing the cost of generation and transmission expansion to account for the yearly increasing demand as well as reducing transmission loss by driving local self-sufficiency \cite{tushar2019grid}.  
 
Due to the  widespread prospect,  P2P  energy trading  mechanism  has raised extensive interest from the research community. A large body of works has made efforts to address the  matching of supply and demand bids  for  prosumers with customized preferences or interests.  
This is usually termed market clearing mechanisms. 
 The mechanisms  in discussion are diverse and plentiful, which can be broadly categorized  by  optimization-based approaches  \cite{baroche2019prosumer, cui2019peer, yang2021optimal}, auction schemes \cite{tushar2019grid, teixeira2021single}, and  bilateral  contract negotiations \cite{morstyn2018bilateral, kim2019p2p}.  Quite  a few of comprehensive and systematic reviews have documented those  market clearing mechanisms, such as \cite{sousa2019peer, tushar2018transforming, khorasany2018market}.  
 On top of that, a line of works has discussed the trust, secure, and transparent  implementation of P2P market scheme  by combing with the well-known blockchain technology, such as    \cite{hamouda2020novel,esmat2021novel}. 

The above studies are mainly focused on the business models of energy trading in  virtual layer and in the shoes of prosumers. 
Whereas the energy exchanges in a P2P market require the  delivery  in   physical layer  taken  by the power grid operator who 
is responsible for  securing  the transmission capacity constraints  and compensating  the transmission loss. 
 In this regard,  the  effective interaction between the prosumers making energy transaction in virtual layer and the power grid operator delivering the trades in  physical layer is essential for the successful deployment of P2P market scheme.  
 The interaction requires to  secure the economic benefit of prosumers  in the P2P  market as well as ensure the operation feasibility of  power grid operator. This has been identified as one key  issue  that remains  to be addressed \cite{tushar2020peer}.

Network charge which allows the grid operator to impose some grid-related cost on  the  prosumers for energy exchanges, has been advocated as a promising tool to bridge this interaction. 
 Network charge is  reasonable and natural considering many aspects.  First of all, network charge is necessary for the power grid to attribute the network investment cost and the transmission loss \cite{baroche2019exogenous}. 
In traditional power systems where customers trade energy with the power grid, such  cost  has  been internalized  in the electricity price, 
it is therefore  natural  to pass the similar cost with P2P to the prosumers via some price mechanisms.  
Besides, network charge can work as a means  to shape the P2P  energy trading market  to ensure the feasible  delivery of trades in physical layer taken by the grid operator  \cite{guerrero2018decentralized}. 
Generally, network charge is charged by the trades,  therefore it can be used to guide the behaviors of the prosumers in the P2P market. 
As a result, several recent works have relied on network charge to account for  the  grid-related cost or shape the P2P markets, such as  \cite{paudel2020peer,  baroche2019exogenous, kim2019p2p}.  Specifically, \cite{paudel2020peer} has involved network charge   in developing  a decentralized P2P market clearing mechanism. 
The work  \cite{baroche2019exogenous} comparatively simulated  three network charge models (i.e., unique  model,  electrical distance based model, and zonal model) on shaping the P2P market.
The work  \cite{kim2019p2p} has relied on a  network charge  model to achieve \emph{ex-post} transmission loss allocations across the prosumers.  The above works have demonstrated that network charge can effectively shape the P2P transaction market.  In addition, network charge can work as a tool to attribute grid-related cost and  transmission loss which are  actually taken by the grid operator.   However, the existing works have mainly focused on studying how the network charge will affect the behaviors of prosumers in a P2P market instead of studying how the  network charge price to be designed which couples   the grid operator and  the prosumers acting as  independent stakeholders and playing different roles.

This paper fills the gap by jointly considering 
 the power grid operator who provides transmission service and the prosumers  who make energy transaction in a P2P market and propose an optimal network charge mechanism.  Particularly,  considering that  the power grid operator and the prosumers are independent stakeholders and have different objectives, we model the interaction between the power grid operator and the prosumers as a Stackelberg game.  First,  the grid operator decides on the optimal network charge price  to trade off the network charge revenue and the transmission loss considering  the network constraints, and then the prosumers optimize their energy management (i.e., energy consuming, storing and trading) for maximum economic benefits. 
Our main contributions are:
 \begin{itemize}
\item[(C1)] We propose a Stackelbeg game model to account for the interaction  between the power grid operator imposing network charge price and the prosumers  making energy transaction  in a P2P market. The distributed renewable generators and energy storage (ES) devices on the prosumer side are considered.   We prove that the Stackelberg game admits an \emph{equilibrium}  network charge price.

\item[(C1)] To deal with the computational challenges of obtaining  the  network charge price,  we convert the bi-level optimization problem yield by the Stackelberg game to a single-level mixed-integer quadratic programming (MIQP) by exploring the problem structures.  The method can  handle a reasonable scale of prosumers  efficiently. 
	
\item[(C2)] By simulating  the IEEE bus systems, we demonstrate that the network charge mechanism  is favorable as it  can benefit both the grid operator and the prosumers for empowering the P2P market. Moreover, it  can provide \emph{near-optimal} social welfare. 
In addition, we find that  the presence of ES will  make the prosumers more sensitive to the network charge price changes. 
\end{itemize}

The rest of this paper is as:  in Section II, we present the Stackelberg  game formulation; in Section III, we propose a single-level conversion method; in Section VI, we examine the proposed network charge mechanism via case studies; in Section V, we conclude this paper and discuss the future work.

\section{Problem Formulation}


\begin{figure}[h]
	\centerline{\includegraphics[width=3.0 in]{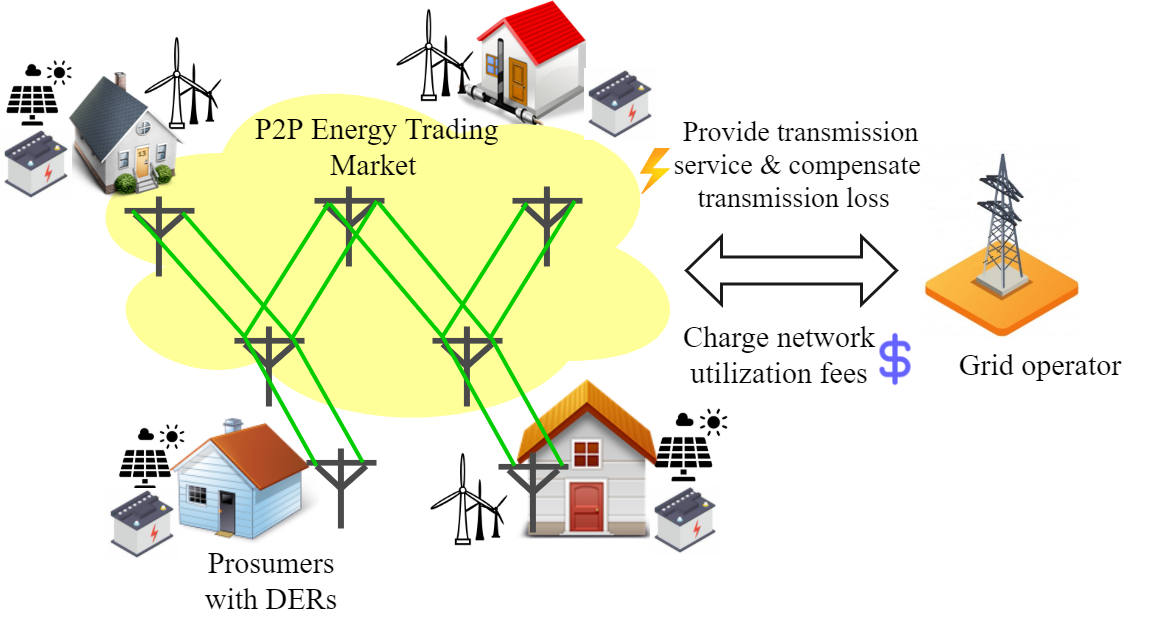}}
	\caption{Interaction between the grid operator and a P2P energy trading market.}
	\label{fig:P2P_market}
\end{figure}

Fig. \ref{fig:P2P_market} shows the interaction between the grid operator and a P2P energy trading market to be discussed in this paper.   
By providing transmission service and compensating the transmission loss for empowering P2P trading, the grid operator plays the leading role by deciding the network charge price.
 In response, the prosumers with DERs (e.g., solar panels, wind turbines,  ES, etc.)  in the P2P market will optimize their optimal energy management (i.e., energy consuming, storing and trading) for maximum economic benefits.   In this paper, we assume  the grid operator and  prosumers are independent stakeholder and  are both profit-oriented,  expecting to maximizing their own profit via the interaction. 
For the grid operator, the profit is evaluated by the network charge revenue minus the cost of transmission loss. 
For the prosumers, the profit is quantified by the utility (i.e., satisfaction) of consuming certain amount of energy  and  the  energy  cost such as  the network charge payment.  The objective of this paper is to determine the  optimal network charge price  that maximizes  the grid profit  while securing the prosumers' profit in the P2P market.

\vspace{-2mm}
\subsection{Network Charge Model}

How to charge  P2P energy trading for network utilization fee is still an open issue.  One way in extensive discussion  is based on the electrical distance and the volume of transaction.  Specifically,  if prosumer  $i$ buys $p_{ij}$ [\si{\kilo\watt}] units of power from prosumer $j$  over an electrical distance of $d_{ij}$ [\si{\kilo\meter}], the network charge is calculated as 
\begin{equation} \label{eq:network_charge_model}
\begin{split}
& T(p_{ij}) = \gamma d_{ij} p_{ij} \\
\end{split}
\end{equation}
where $\gamma$[s\$/(\si{\kilo\watt \cdot \kilo\meter})] is the network charge price determined by the grid operator, which represents the  network utilization fee for per unit of energy transaction over per unit of electrical distance.

The electrical distance is determined by the electrical network topology and the measures used.   For a given electrical network,   there are several popular  ways  to  measure  the electrical distances as discussed in  \cite{cuffe2015visualizing}.  One of them is the  \emph{Power Transfer Distance Factor} (PTDF) 
which  has been mostly used  for network charge calculations (see  \cite{paudel2020peer,  baroche2019exogenous, kim2019p2p} for examples). 
We therefore use  the PTDF for measuring the electrical distances.  
 For an electrical network characterized by transmission lines $\LL$,   the electrical distance between any trading peers $i, j$ based on PTDF  is  defined as  
\begin{align} \label{eq:PTDF}
d_{ij} = \sum_{\ell \in \LL} \vert {\rm PTDF}_{\ell, ij}\vert
\end{align}
where ${\rm PTDF}_{\ell, ij}$ represents the PTDF  of prosumer $i, j$ related to transmission line $\ell \in \LL$, which characterizes the  estimated  power flow change of line $\ell$ caused by  per unit of  energy transaction between prosumer $i$ and prosumer $j$ according to the DC power flow sensitivity analysis. 

PTDF is directly derived  from  the DC power flow equations and the details can be found in  \cite{christie2000transmission}.  In the following, we only summarize the main calculation procedures.   For an electrical network characterized by $N$ buses and $L$ transmission lines, we first have the nodal acceptance  matrix: 
\begin{align*}
	B_{ij}  = \begin{cases}
		\sum_{k = 1}^N \frac{1}{x_{ik}}, & {\rm if}~j = i. \\
		-\frac{1}{x_{ij}}, &{\rm if}~ j \neq i. \\
	\end{cases}
\end{align*}
where $x_{ij}$ represents the reactance of the  line connecting bus $i$ and bus $j$. 

We denote $\B_r$ as the sub-matrix of $\B$ which eliminates the row and column related to the reference bus $r$. Without any loss of  generality, we specify  bus $N$ as  the reference bus, we therefore have  $\B_r = \B[1: N-1, 1: N-1]$ and the reverse  $\X _r= \B_r^{-1}$.  By setting \emph{zero} row and column for the reference bus $r = N$, we have the augmented matrix:
\begin{align*}
	\X = \begin{pmatrix}
		\X_r & \mathbf{0}\\
		0  & 0 
	\end{pmatrix}
\end{align*}

By using matrix $\X$, we can calculate the PTDF by 
\begin{align}\label{eq:PTDF1}
	{\rm PTDF}_{\ell, ij} = \frac{X_{mi} - X_{mj} - X_{ni} + X_{nj}}{x_{\ell}}
\end{align}
where $X_{mi}, X_{mj}, X_{ni}, X_{nj}$ represent the elements of matrix $\X$ at row $m, n$ and column $i, j$, $\ell$ is the transmission line connecting bus $m$ and bus $n$.  

\textbf{An illustration example}: we use the 5-bus system in Fig. \ref{fig:IEEE_5_bus} to  illustrate the interpretation of electrical distances based on PTDF. 
Based on  \eqref{eq:PTDF}-\eqref{eq:PTDF1} and the reactance parameter $\x$, we can obtain the electrical distance $\mathbf{d}$ shown in Fig. \ref{fig:IEEE_5_bus} (b).  Particularly,  we have the electrical distance between bus $1$ and bus $3$:  $d_{13} = 0.2958 + 0.4930 + 0.2113 + 0.2958 + 0.2113 = 1.5072$ which are the PTDF of bus $1$ and bus $3$ related to the 5 transmission lines.  As shown in Fig. \ref{fig:IEEE_5_bus} (a), the PTDF  for bus $1, 3$ can be 
 interpreted as the total power flow changes  of all transmission lines caused by  per unit of energy transaction between the bus $1$ and $3$ according to the DC power flow analysis.
 \begin{figure}[h]
 	\centerline{\includegraphics[width=3.0 in]{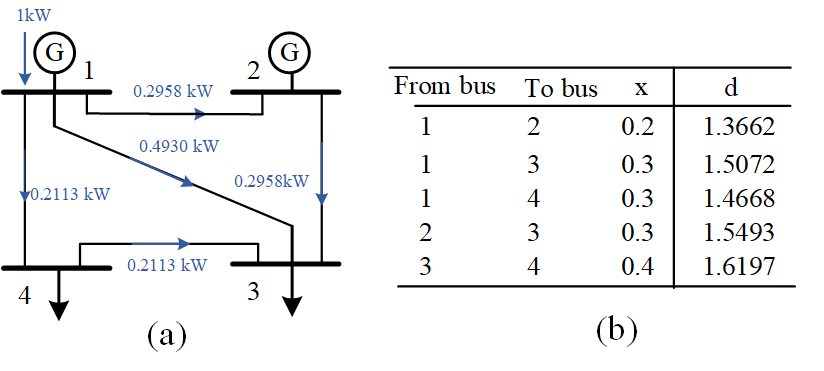}}
 	\vspace{-3mm}
 	\caption{(a) Power flow changes of  all transmission lines  if  bus $1$ transfers  1 \si{\kilo\watt} power to bus $3$ based on DC power flow analysis.  (b) The  electrical distances between the buses based on PTDF for the 5-bus system.  }
 	\label{fig:IEEE_5_bus}
 \end{figure}

\subsection{Stackelberg Game Formulation}

As discussed,  the interaction between the grid operator and the prosumers shows  a hierarchical structure.  
This corresponds well to a  Stackelberg game where   the power grid behaves as  the \emph{leader} and the prosumers are  \emph{followers}.   Before the formulation, we first define the main notations in TABLE \ref{tab:notations}. 

\vspace{-3mm}
\begin{table}[h]
	\setlength\tabcolsep{3pt}
	\centering
	\caption{Main notations}
	\vspace{-2mm}
	\label{tab:notations}
	\begin{tabular}{ll}     
		\toprule
		Notation  & ~~~~~~~~~~~~~~~~~~~Definition \\ 
		\hline
		$i, j$         & Prosumer/bus index.   \\
		$t $           &  Time index.     \\
		$\gamma$  & Network charge price.  \\
		$p_{ij, t}^{+}/p_{ij, t}^{-}$  &   \makecell[l]{Traded (buy or sell) energy between prosumer $i, j$. }\\ 
		$\theta_{i, t}$ & Phase angle at bus $i$.   \\
		$\PP_{i, t}$   & Consumed or generated power of prosumer $i$.\\
		$P_{i, t}$  & Injected power at bus $i$.\\
		$p_{i, t}^{\rm ch}/p_{it}^{\rm dis}$ &  Charged/discharged  power of prosumer $i$'s ES.  \\
		$ e_{i, t}$   & Stored energy of prosumer $i$'s ES. \\
		$U_{i, t}(\PP_{i, t})$  & Utility function of prosumer $i$. \\
		$T(p_{ij, t}^{+})$   & Network charge for trading $p_{ij, t}^{+}$ units of  energy. \\ 
		$F_{ij}^{\max}$   & Transmission network capacity for line $(i, j) \in \LL$. \\
		$p_{i,t}^{\rm r}$  & Renewable generation of prosumer $i$. \\
		$C_{ij}^{\max}$  & Max. trading power between prosumer $i, j$. \\
		$e_i^{\min}/e_i^{\max}$ & Min./max. stored energy  of  prosumer $i$'s ES.\\
		$\PP_{i, t}^{\min}/\PP_{i, t}^{\max}$   & Min./max. consumption/generation of prosumer $i$.  \\
		\bottomrule 
	\end{tabular}
\end{table}

\subsubsection{Leader} 
In the upper level, the power grid optimizes the network charge price $\gamma$ to trade off the network charge   revenue and the transmission loss  considering the transmission network constraints.   Network  charge revenue is calculated   by \eqref{eq:network_charge_model} and  the power transmission loss is  consolidated by the DC power flow of  the transmission  network  \cite{ding2018distributed}. We have the problem for the power grid:
\begin{subequations} 
	\begin{alignat}{4}
\label{eq:upper-level}\min_{\x_U} ~&{\rm Profit}=  \sum_{t }\sum_i  \sum_j \big( T(p_{ij, t}^{+})  + T(p_{ij, t}^{-}) \big)/2 \tag{${\rm P}_U$} \\
& \quad \quad~~- \rho \sum_t \sum_{(i,j) \in \LL} b_{ij} (\theta_{i, t}- \theta_{j, t}) ^2 \notag\\
{\rm  s.t.} &~\gamma_{\min} \leq \gamma \leq \gamma_{\max}.  \\
\label{eq:7b}& \B  \bm{\theta}_t = \mathbf{P}_t, \forall t. \\
\label{eq:7c}& \theta_{r, t}= 0, \forall t.   \\
\label{eq:7d}& \vert (\theta_{i, t} - \theta_{j, t}) b_{ij}\vert \leq F_{ij}^{\max}, \forall (i, j) \in \LL, t. \\
\label{eq:7e} &P_{i, t} = \smallsum_{j} p_{ij, t}^{-} - \smallsum_{j}  p_{ij}^{+}, \forall i, t. \\
\label{eq:7f} & \mathbf{P}^{\min} \leq  \mathbf{P}_t \leq \mathbf{P}^{\max}, \forall t. 
	\end{alignat}
\end{subequations}
where the decision variables for the power grid operator are  $\x_{\rm U} = [\gamma, \theta_{i, t}], \forall i, t $.  We use $\mathbf{P}_t = [P_{i, t}], \forall i$ to denote the  power injections at the buses and  $b_{ij}$ denotes the admittance of the line connecting bus $i$ and bus $j$.   We have $ \gamma_{\min}, \gamma_{\max} > 0$ characterize  the range of  network charge price.  We use the term $\rho \smallsum_{(i,j) \in \LL} b_{ij} (\theta_{i, t} - \theta_{j, t}) ^2 = \rho \smallsum_{(i, j) \in \LL}  P_{ij, t}^2/b_{ij} $ related to the power flows to quantify the consolidated transmission  loss over the transmission networks $\LL$   and  $\rho$ is the transmission  loss cost coefficient \cite{ding2018distributed}.
Constraints \eqref{eq:7b} represent the DC power flow equations.  Constraints \eqref{eq:7c} specify  the  phase angle of  reference bus $r$. 
Constraints \eqref{eq:7d} model the  transmission line capacity limits. In this paper, we  use  the DC power flow model  to account for the transmission constraints and transmission loss. Whereas the proposed framework  can be  readily extended to  AC power flow model by replacing \eqref{eq:7b}-\eqref{eq:7d} with the DistFlow \cite{farivar2013branch} or the modified DistFlow \cite{rigo2022iterative} model. The  nonconvex AC power flow model can be further convexified  into a second-order cone program (SOCP) or a semi-definite program (SDP). Then the proposed method of  this paper can still be used to solve the  problem though with increased problem complexity.



\subsubsection{Followers}   In the lower level, the prosumers  in the P2P market will  respond to the network charge price $\gamma$ for maximal economic benefit. 
We use  $U_{i, t}(\PP_{i, t})$ to represent the utility functions of prosumer $i$. Due to the presence of DERs, a prosumer could be a consumer or a producer. 
In this  regard, $U_{i, t}(\PP_{i, t})$  could represent the satisfaction of a customer for consuming $\PP_{i, t}$ units of power or the   cost  of a producer for generating $\PP_{i, t}$ units of energy.  
We also involve the distributed renewable generators and  ES devices on the prosumer side in the formulation. 
In this paper, we assume the prosumers will cooperate with each other in the P2P market  and formulate the problem as a centralized optimization problem as many existing works have proved that the cooperation can make all  prosumer better off with some suitable \emph{ex-post} profit allocation mechanisms  (see \cite{yang2021optimal, jo2020demand, rey2018strengthening} for examples).
Since the network charge is measured by the traded power regardless of the direction, we distinguish the purchased power and sold power between  prosumer $i$ and prosumer $j$ by $p_{ij, t}^{+}$ and $p_{ij, t}^{-}$. 
The problem to optimize the total prosumer profit  considering  network charge payment  is presented  below.  
  \begin{subequations} 
	\begin{alignat}{4}
\label{eq:lower-level}\max_{\x_L}&~~{\rm Profit} = \sum_t \sum_{i} U_{i, t}(\mathcal{P}_{i, t})  \tag{${\rm P}_L$} \\
& \quad \quad \quad - \sum_t \sum_{i} \sum_{j} \big( T(p^{+}_{ij, t}) + T(p^{-}_{ij, t})\big)/2 \notag\\
\label{eq:8a}\text{s.t.}~~& p_{ij, t}^{+} = p^{-}_{ji, t},  ~~~\forall i, j, t.\\
\label{eq:8b} & 0 \leq p_{ij, t}^{+} \leq C_{ij}^{\max}, ~~\forall i, j, t. \\
\label{eq:8c} & 0 \leq p_{ij, t}^{-} \leq C_{ij}^{\max}, ~~\forall i, j, t. \\
\label{eq:8d}& \PP_{i, t}\leq  p_{i, t}^{\rm r}  \!+ \!p_{i, t}^{\rm dis} \!-\! p_{i, t}^{\rm ch} \!+\! \smallsum_{j} p_{ij, t}^{+} \!-\! \smallsum_j p_{ij, t}^{-}, \forall i, t.\\
\label{eq:8e}& \PP_{i, t}^{\min} \leq \PP_{i, t}\leq \PP_{i, t}^{\max}, ~\forall i , t. \\
\label{eq:8f}& e_{i, t + 1} = e_{i, t} + p_{i, t}^{\rm ch} \eta - p_{i, t}^{\rm dis}/\eta,  ~\forall i, t. \\
\label{eq:8g}& 0 \leq p_{i, t}^{\rm ch}  \leq P_i^{\rm ch, \max}, ~~\forall i, t. \\
\label{eq:8h}& 0 \leq p_{i, t}^{\rm dis} \leq  P_i^{\rm dis, \max},  ~~\forall i, t. \\
\label{eq:8i}& e_i^{\min} \leq e_{i, t} \leq e_i^{\max},  \forall i, t. 
	\end{alignat}
\end{subequations}
where  the decision variables for the prosumers are  $\x_L = [p_{ij, t}^{+}, p_{ij, t}^{-}, \PP_{i, t}, p_{i, t}^{\rm ch}, p_{i, t}^{\rm dis}, e_{i, t}], \forall i, t$. 
Constraints \eqref{eq:8a} model the consistence of energy transaction  between the sellers and the buyers. 
Since the transmission loss is compensated by the power grid operator, we have the amount of energy that  prosumer  $i$ buys from prosumer $i$ equals that  prosumer $j$ sells to prosumer $i$. 
Constraints \eqref{eq:8b}-\eqref{eq:8c} impose  the transaction limits between the trading peers. 
Constraints \eqref{eq:8d}  ensure  the load balance of each prosumer. Particularly, we use inequality to capture the case where some renewable generation is  curtailed. 
Constraints \eqref{eq:8e} characterize the  demand or supply flexibility of  the prosumers. Constraints \eqref{eq:8f} tracks the stored energy of prosumers' ES  with $\eta \in (0, 1)$ denoting the charging/discharging efficiency. 
Constraints \eqref{eq:8g}-\eqref{eq:8h} impose the charging, discharging and stored  energy capacity limits. 
In this paper, we focus on the energy trading among the prosumers in the P2P market.  For the case where the prosumers also trade electricity with the power grid, the proposed model  can be readily  extended by adding the cost or revenue related to the energy trading with the grid to the prosumers' objective in the lower-level problem \eqref{eq:lower-level}.

\subsubsection{Piece-wise linear utility function} 
\begin{figure} [h]
	\centering
  \includegraphics[width = 3.0 in]{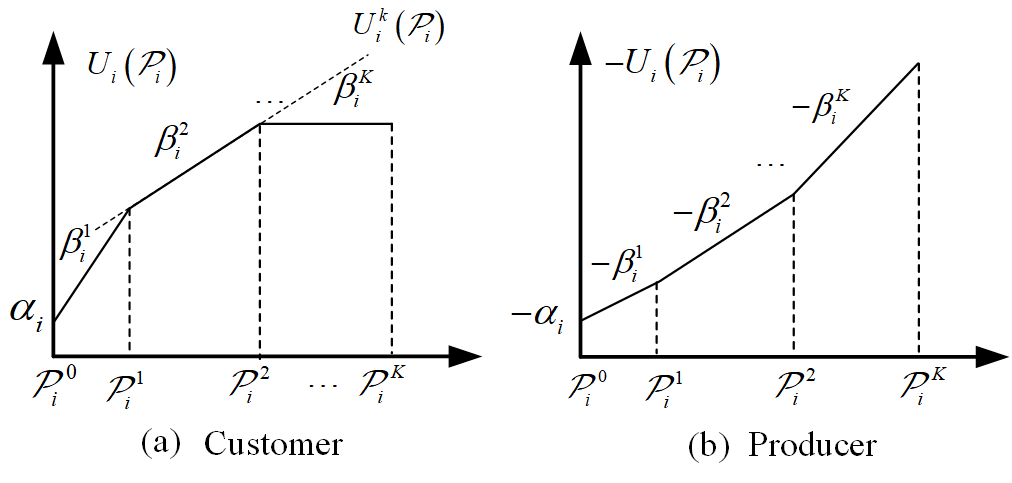}
	\caption{(a) Piece-wise linear (PWL) utility function for a consumer $\nabla U_{i}(\PP_{i}) \geq 0$. (b) Piece-wise  linear utility  (PWL) function for a producer $\nabla U_{i}(\PP_{i}) \leq 0$ (time $i$ is omitted). }
	\label{fig:PWL-utility}
	\vspace{-3mm}
\end{figure}

This paper employs  concave piece-wise linear (PWL) utility  functions to capture the prosumers' demand or supply flexibility as shown in Fig. \ref{fig:PWL-utility}.  The motivation behind is that PWL functions are universal and can  approximate  all types of  utility functions, such as quadratic and logarithmic \cite{XuGuoGao}. We may obtain the PWL utility functions by linearizing non-linear utility functions   or directly learn it from data \cite{wu2011tighter}.  
Due to the presence of DERs, the prosumer could be a consumer in energy deficiency or a producer with energy surplus. This could be universally formulated by the PWL utility function but with the opposite sign of the slopes. We use Fig. \ref{fig:PWL-utility} (a) and (b) to show the two scenarios (time $i$ is omitted):  if the slope of the PWL utility function  is non-negative $U_{i} (\PP_i) \geq 0$,  the prosumer plays the role of customer and the prosumer will play the role of producer if $U_i(\PP_i) \leq 0$.  As shown in Fig \ref{fig:PWL-utility}, a general  PWL utility function  composed of $K$ segments is characterized by  the transition points and slopes: $\PP^k_i$ and   $\beta^k_i, k = 1, 2, \cdots, K$.   The function associated with the $k$-th segment can be described as 

{\small 
\begin{align}\label{eq:utility}
U^k_i(\PP_i) \!= &\alpha_{i} \!+\!  \sum_{\ell = 1} ^{k-1} \beta^{\ell}_i\left( \PP^{\ell}_i \!-\! \PP^{\ell-1}_i \right)  \!\!+\!\! \beta^k_i \left( \PP_i \!-\! \PP^{k-1}_i\right), \forall i, k. 
\end{align} }
where $\alpha_{i}$  is the  constant component of prosumer $i$'s  utility function,  which could represent the  satisfaction level of a prosumer for consuming zero unit of energy or the start-up generation cost  for a producer.  It is easy to note that we  have $U_i(\PP_i) = U_i^k(\PP_i)$ if $\PP_i \in [\PP^{k-1}_i, \PP^k_i)$. 


For the proposed Stackelberg game, we have the following results regarding the existence of  \emph{equilibrium}. 
\begin{theorem} \label{them:theorem1}
The Stackelberg game \eqref{eq:upper-level}-\eqref{eq:lower-level} admits an  equilibrium. 
\end{theorem}
\begin{proof}
For the lower-level problem \eqref{eq:lower-level},  we note  that the problem is compact and convex with any given network charge price $\gamma$.  This implies that  the optimal solution for the lower-level problem \eqref{eq:lower-level} always exists and can be expressed  by $\x_L(\gamma)$.  
By substituting the closed-form solution $\x_L(\gamma)$ (if explicitly available) into the upper-level problem \eqref{eq:upper-level} and by expressing the phase angle decision variables $\bm{\theta}$ with  the power flows  determined by  the lower-level problem solution $\x_L(\gamma)$,  we can conclude a single-level optimization problem for the Stackelberg game with the only bounded decision variables $\gamma \in [\gamma^{\min}, \gamma^{\max}]$,  which will yield at least one optimal solution. This implies that the proposed Stackelberg game adopts  at least one Stackelberg \emph{equilibrium}. 
\end{proof}

\begin{remark}
	The existence of the Stackelberg \emph{equilibrium} implies that the proposed optimal network charge model can yield  an optimal network charge price  that maximizes the profit of the grid operator while considering the cost-aware behaviors of the prosumers in the P2P market. 
\end{remark}

\section{Methodology}
Note that the optimal network charge associates with the \emph{equilibrium} of the Stackelberg game  \eqref{eq:upper-level}-\eqref{eq:lower-level} which yields a bi-level optimization. 
Bi-level optimization  is generally   NP-hard  and computationally intensive \cite{dempe2020bilevel}.  
This section proposes a  method to convert   the bi-level problem to  a  single-level  problem that can  accommodate  a  reasonable scale of  prosumers by exploring the problems  structures. 
To achieve this goal, we first restate the lower-level problem \eqref{eq:lower-level} as  

{\small 
\begin{subequations} 
	\begin{alignat}{4}
\label{eq:reformulated_lower_level}\max_{\x_L} &~{\rm Profit}=  \sum_t \sum_{i } u_{i, t} - \sum_t \sum_{i} \sum_{j} T(p_{ij, t}^{+})\tag{${\rm P}^{'}_L$}  \\
\label{eq:10a}{\rm s.t.}~  & 0 \leq p_{ij, t}^{+} \leq C_{ij}^{\max}:~ \quad \quad ~~\ubar{\nu}_{ij, t},  \bar{\nu}_{ij, t} \geq 0, ~\forall i, j, t. \\
& \PP_{i, t} \!\leq\!  p_{i, t}^{\rm r}  \!+\! p_{i, t}^{\rm dis} \!-\! p_{i, t}^{\rm ch} \!+\! \smallsum_{j} p_{ij, t}^{+} \!-\!\!\smallsum_j p_{ji, t}^{+}: \notag\\
\label{eq:10b}&\quad \quad \quad \quad \quad \quad \quad \quad \quad \quad \quad \quad \quad \quad~~\mu_{i, t} \geq 0, \forall i, t. \\
\label{eq:10c}& \PP_i^{\min} \leq \PP_{i, t}\leq \PP^{\max}_i:~ \quad \quad~~ \ubar{\sigma}_{i, t}, \bar{\sigma}_{i, t} \geq 0, ~\forall i, t.  \\
\label{eq:10d}& e_{i, t + 1} = e_{i, t} + p_{i, t}^{\rm ch} \eta - p_{i, t}^{\rm dis}/\eta: \quad  ~\mu_{i, t}^{\rm e} \in \R,  ~\forall i, t. \\
\label{eq:10e}& 0 \leq p_{i, t}^{\rm ch}  \leq P_i^{\rm ch, \max}:  \quad \quad \quad   ~~\ubar{\mu}^{\rm ch}_{i, t},  \bar{\mu}_{i, t}^{\rm ch} \geq 0, ~\forall i, t.\\
\label{eq:10f}& 0 \leq p_{i, t}^{\rm dis} \leq  P_i^{\rm dis, \max}:   \quad \quad  ~~~ \ubar{\mu}_{i, t}^{\rm dis},  \bar{\mu}_{i, t}^{\rm dis} \geq 0,  ~\forall i, t. \\
\label{eq:10g}&  e_i^{\min} \!\leq\! e_{i, t} \leq  e_i^{\max}: \quad  \quad \quad  ~~~ \ubar{\mu}_{i, t}^{\rm e},  \bar{\mu}_{i, t}^{\rm e} \geq 0,  ~\forall i, t. \\
\label{eq:10h}& u_{i, t} \leq U_i^k(\PP_{i, t}): \quad  \quad \quad \quad  \quad  \quad\delta_{i, k, t} \geq 0, ~\forall i, k, t.  
	\end{alignat}
\end{subequations}}
where  the decision variable $\mathbf{P}^{-} = [p_{ij, t}^{-}], \forall i, j, t$  are removed based on  $p_{ij, t}^{+} = p_{ji, t}^{-}, \forall i, j, t$. Besides,  some auxiliary variables  $u_{i, t}$ are introduced to relax the non-smooth prosumer utility functions. Since the utility function is concave, it is easy to prove  that   \eqref{eq:reformulated_lower_level} is equivalent to \eqref{eq:lower-level}.  Additionally,   the  dual variables  for  the constraints are defined  the right-hand side.

For the reformulated lower-level problem \eqref{eq:reformulated_lower_level}, we can draw the  Karush–Kuhn–Tucker (KKT) conditions \cite{boyd2004convex}. 
We first have the first-order optimality conditions:   
\begin{subequations} 
	\begin{alignat}{4}
\label{eq:11a}& \partial L/\partial \PP_{i, t}  \!=\! \mu_{i, t} \!-\! \ubar{\sigma}_{i, t}  \!+\! \bar{\sigma}_{i, t}  \!-\! \smallsum_k \delta_{i, k, t} \nabla U_{i,t}^k(\PP_{i, t}) \!=\! 0\\
& \partial \mathbb{L}/\partial p_{ij, t}^{+} \!=\!\gamma d_{ij} - \ubar{\nu}_{ij, t}  + \bar{\nu}_{ij, t} - \mu_{i, t} + \mu_{j, t}  = 0\\
& \partial \mathbb{L}/\partial u_{i, t} \!=\!  - 1  +  \smallsum_k \delta_{i, k, t} = 0 \\
&\partial \mathbb{L}/\partial p_{i, t}^{\rm ch} \!=\! \mu_{i, t} -  \mu_{i, t}^{\rm e} \eta- \ubar{\mu}_{i, t}^{\rm ch} =0\\
& \partial \mathbb{L}/\partial p_{i, t}^{\rm dis} \!=\! - \mu_{i, t} + \mu_{i, t}^{\rm e}/\eta - \ubar{\mu}_{i, t}^{\rm dis} + \bar{\mu}_{i, t}^{\rm dis} = 0\\
\label{eq:11f}& \partial L/\partial e_{i, t} \!=\! - \mu_{i, t}^{\rm e} \!+\! \mu_{i, t-1}^{\rm e}  \!-\! \ubar{\mu}_{i, t-1}^{\rm e}  \!+\! \bar{\mu}_{i, t-1}^{\rm e}  = 0, \forall  t > 1
	\end{alignat}
\end{subequations}
where we use $\mathbb{L}$ to denote the Lagrangian function associated with  \eqref{eq:reformulated_lower_level}. 
Based on \eqref{eq:utility}, we have $\nabla U_{i, t}^k(\PP_{i, t}) = \beta_i^k$ which represents the slope of the prosumer $i$'s  utility function at the $k$-th segment.

In addition,  we have the complementary constraints for the inequality constraints \eqref{eq:10a}-\eqref{eq:10h}. Using \eqref{eq:10d} as an example, we have the complementary constraints:
\begin{align} \label{eq:complementary}
\mu_{i, t}\big( \PP_{i, t} \!-\!  p_{i, t}^{\rm r}  \!-\! p_{i, t}^{\rm dis} \!+\! p_{i, t}^{\rm ch} \!-\! \smallsum_{j} p_{ij, t}^{+} \!+\!\!\smallsum_j p_{ji, t}^{+}\big) \!=\!0, \forall i, t.
\end{align}

The general way to handle the non-linear complementary constraints is  to introduce binary variables to relax the constraints (see \cite{feng2020stackelberg} for an example). 
This could be problematic for problem \eqref{eq:reformulated_lower_level} due to the large  number of inequality constraints. 
To deal with the computational challenges, we make use of  the linear programming (LP) structure of  problem \eqref{eq:reformulated_lower_level}.  For a LP,  we have the  \emph{strong duality} and the \emph{complementary constraints} are  interchangeable (see \cite{Bradley1977}, Ch4, pp. 147 for detailed proof). Therefore, we use the strong duality condition for problem \eqref{eq:reformulated_lower_level} to replace the complementary constraints, such as \eqref{eq:complementary}.  
We have  the strong duality for problem \eqref{eq:reformulated_lower_level}: 

{\small 
\begin{equation} \label{eq:strong_duality}
\begin{split}
&\!-\! \sum_t \sum_i \sum_j  \bar{\nu}_{ij, t} C_{ij}^{\max}\! - \!\sum_t \sum_i \mu_{i, t} p_{i, t}^{\rm r} \!\!+\!\!\sum_t \sum_i \ubar{\sigma}_{i, t} \PP_{i, t}^{\min} \\
&\!-\!\!\sum_t \!\sum_i\bar{ \sigma}_{i, t} \PP_{i, t}^{\max} \!\!-\!\!\sum_t\! \!\sum_i  \bar{\mu}_{i, t}^{\rm ch} P_i^{\rm ch, \max} \!\!-\!\! \sum_t \!\sum_i  \bar{\mu}_{i ,t}^{\rm dis} P_i^{\rm dis, \max} \\
& \!\!+\!\!\sum_t \sum_i  \ubar{\mu}_{i, t}^{\rm e}  e_i^{\min} 
\!\!-\!\!\sum_t \sum_i \bar{\mu}_{i, t}^{\rm e} e_i^{\max} \!\!-\! \!\sum_t \sum_i \sum_k \delta_{i, k, t} U_i^k(0) \\
& = \sum_t \sum_i \sum_j T(p_{ij, t}^{+}) - \sum_t \sum_i u_{i ,t} 
\end{split}
\end{equation}
}


Note that the strong duality \eqref{eq:strong_duality} can be used to eliminate the large number of  non-linear complementary constraints  but requires to tackle the bi-linear terms related to the network charge calculations:    $T(p_{ij, t}^{+}) =\gamma d_{ij}  p_{ij, t}^{+}$. To  handle such bi-linear terms,  we discretize  the network charge price and convert the non-linear terms into mixed-integer constraints. Specifically, we first define an auxiliary variable $Z$:
$$ Z = \sum_t\sum_i \sum_j  d_{ij} p_{ij, t}^{+}$$

We thus  have the  total network charge for  P2P transaction: 
\begin{equation} \label{eq:non-linear}
\begin{split}
\sum_t\sum_i \sum_j T(p_{ij, t}^{+})  = \gamma Z
\end{split}
\end{equation}

We discretize the range of  network charge price   $[\gamma_{\min}, \gamma_{\max}]$   into $L$ levels $\{ \gamma_{1},  \gamma_2, \cdots, \gamma_L\}$ with an equal  interval $\Delta \gamma =  (\gamma_{\max} -\gamma_{\min})/L$.
Accordingly,  we introduce the binary variables $\mathbf{x} = [ x_\ell ], \ell = 1, 2, \cdots, L$ to indicate which  level of network charge price is selected, we thus have 
\begin{align}
& \gamma Z = \smallsum_{\ell =1} ^L x_{\ell} \gamma_{\ell} Z  \\
\label{eq:linearization}&\smallsum_{\ell =1}^L x_{\ell} =1, ~x_{\ell} \in \{0, 1\}
\end{align}

Note that the network charge calculations rely on the product of binary variable $x_{\ell}$ and continuous variable $Z$. This can be 
equivalently expressed by  the integer algebra:
\begin{align} 
\label{eq:integer_algebra1}&\quad \quad ~ -M x_{\ell} \leq Y_{\ell} \leq M x_{\ell} \\
\label{eq:integer_algebra2}& -M(1-x_{\ell})  \leq Z - Y_{\ell} \leq M(1-x_{\ell})
\end{align}
where we have $\gamma Z  = \sum_{\ell=1}^L \gamma_{\ell} Y_{\ell}$ and  $M$ is a sufficiently large positive constant.

By plugging $\smallsum_t \smallsum_i \smallsum_j T(p_{ij, t}^{+})  = \gamma Z=\smallsum_{\ell=1}^L \gamma_{\ell}Y_{\ell}$ in \eqref{eq:strong_duality}, and by 
replacing the lower level problem \eqref{eq:reformulated_lower_level} with KKT conditions,   we have the following  single-level mixed-integer quadratic programming (MIQP): 
\begin{subequations}
	\begin{alignat}{4}
\label{pp:single-level} \max_{\x_U, \x_L, \bm{\lambda}}~&\text{Profit} =  \sum_{\ell = 1}^L \gamma_\ell Y_\ell -  \rho \!\!\!\sum_{(i,j) \in \LL} \!\! \! b_{ij} (\theta_i - \theta_j) ^2 \tag{$P$}\\
& \text{s.t.} ~\eqref{eq:10a}-\eqref{eq:10h}.~~~~~~~~\text{Primal constraints} \notag \\
& ~~~~\eqref{eq:11a}-\eqref{eq:11f}.  ~~~~~~~~~\text{KKT conditions}  \notag \\
&~~~~ \eqref{eq:strong_duality}, \eqref{eq:linearization}-\eqref{eq:integer_algebra2}~~~~\text{Strong~duality} \notag 
		\end{alignat}
\end{subequations}
where $\lambda = [\bm{\ubar{\nu}},  \bm{\bar{\nu}}, \bm{\mu}, \bm{\ubar{\sigma}},  \bm{\bar{\sigma}},  \bm{\mu}^{\rm e},  \bm{\ubar{\mu}}^{\rm ch},  \bm{\bar{\mu}}^{\rm ch}, \bm{\ubar{\mu}}^{\rm dis},  \bm{\bar{\mu}}^{\rm dis}, \bm{\ubar{\mu}}^{\rm e},  \bm{\bar{\mu}}^{\rm e}, \bm{\delta}]$ are the dual variables. 
Note that this single-level conversion favors computation   as   the number of binary variables ($L$)  is  only determined by the granularity of  network charge discretization and independent of the scale of prosumers, making it possible to accommodate a reasonable scale of prosumers.

\section{Case Studies}
In this section,  we evaluate the  performance of the  proposed network charge mechanism via simulations.  We first use  IEEE 9-bus system to evaluate the effectiveness of the solution method, the existence of \emph{equilibrium} network charge price,  and the social welfare.
We further evaluate the performance  on the larger electrical networks including  IEEE 39-bus, 57-bus,  and 118-bus systems.  Particularly, we compare the results  with and without ES on the prosumer side in the case studies. 

\begin{table}[h]
	\setlength\tabcolsep{3pt}
	\centering
	\caption{Simulation set-ups}
	\label{tab:simulation_setup}
	\begin{tabular}{lll}     
		\toprule
		Param. &  Definition~~~~~ & Value \\
		\hline
		$T$  & Time periods & 24  \\
		$\alpha_{i, t}$  & Proumer PWL utility constant & 0 \\
		$\beta_{i, t}^k$  & Prosumer PWL utility slopes & $[0, 1]$ \\
		$K$     & Prosumer PWL utility segments  & 2 or 3 \\
		$[\gamma^{\min}, \gamma^{\max}]$  & Network charge price range & [0, 1] $\si{s\$/(\kilo\watt\cdot \kilo\meter)}$ \\
		$\Delta \gamma$ &   Network charge price discretization   & 0.02 $\si{s\$/(\kilo\watt\cdot \kilo\meter)}$ \\
		$L$  & Network charge discretization levels & 51 \\
		$e_i^{\min}/e_i^{\max}$  & Min./max storaged energy of  ES  & 0/60 \si{\kilo\watt\hour}\\
		$P_i^{\rm ch, \max}$  & Max.  charging  power & 50 \si{\kilo\watt} \\ 
		$P_i^{\rm dis, \max}$  & Max. discharging power & 50 \si{\kilo\watt} \\ 
		$\eta$ & Charging/discharging efficiency & 0.9 \\
		$\rho$ & Transmission loss cost  coefficient & 0.01  \\
		\bottomrule
	\end{tabular} 
\end{table}

\subsection{Simulation Set-ups}
We set up the case studies by rescaling the  real building load profiles  \cite{BuildingDemand} and the renewable  generation  profiles (i.e., wind and solar)  \cite{WindSolar}.  To capture the demand flexibility,  we set the lower prosumer demand as $\PP_{i,t}^{\min} =0$ (we  focus on the flexible demand) and the upper prosumer demand  as  $\PP_{i, t}^{\max}=\text{\emph{demand profile}}_{i, t} +30$ \si{\kilo\watt}. 
For each time period $t$, we uniformly generate the slopes of   prsumer PWL utility functions  in  $\beta_{i, t}^k \in [0, 1]$ with $K= 2$ or $3$ segments (we only consider customers in the following studies and the producers can be included by setting $\beta_{i, t}^k \in [-1, 0]$ if exist).   We set the constant components of PWL utility function as $\alpha_{i, t} = 0$ for all customers. 
 Correspondingly, we equally divide the ranges of prosumer demand $[\PP^{\min}_{i, t},  \PP^{\max}_{i, t}]$ into $K=2$ or $3$ segments to obtain the PWL utility function  transition points $\PP^k_{i, t}$. 
 We simulate the P2P market for 24 periods with a decision interval of one hour.  The  settings for the above parameters and  the prosumers'  ES   are gathered  in TABLE \ref{tab:simulation_setup}.  Particularly,  we set the range of network charge price as  $\gamma^{\min} = 0$ and $\gamma^{\max} = 1.0$ $\si{s\$/(\kilo\watt\cdot \kilo\meter)}$ and the discretization interval  as $\Delta \gamma = 0.02$ $\si{s\$/(\kilo\watt\cdot \kilo\meter)}$  based on the simulation results of Section IV-B, which  suggest  such settings are expected to provide solutions with sufficiently high accuracy.   Besides, the electrical distances  measured by PTDF for the concerned bus systems are directly obtained with the method in Section II-A.

In this paper,  we refer to the P2P market with the proposed network charge as \emph{Optimal P2P}. The  network charge price is obtained by solving problem \eqref{pp:single-level} with the \emph{off-the-shelf} solvers. 
In the following studies, we compare \emph{Optimal P2P} with \emph{No P2P} (P2P transaction is forbidden), \emph{Free P2P} (P2P transaction is allowed without any network charge form the prosumers) and \emph{Social P2P} (P2P  transaction is determined by maximizing the social profit which is the sum of  grid operator profit and prosumer profit defined in \eqref{eq:upper-level} and \eqref{eq:lower-level}). Note that the network charge with \emph{Social P2P} will be internalized  as the grid operator and the prosumers are unified as a whole.   For \emph{Free P2P}, the grid operator has no manipulation on the P2P market and  the optimal transaction can be determined by directly solving the lower-level problem \eqref{eq:lower-level} by removing the network charge components  (To ensure the uniqueness of the solution, we keep the network charge but set a sufficiently small  value).  
 In addition, we examine  the different markets  without and with ES on the prosumer side. For the case without ES, we set $e_i^{\max}, P_i^{\rm ch, \max}, P_i^{\rm dis, \max}$  as \emph{zero}. For the case with ES, we assume each prosumer has a ES with the configurations shown in TABLE \ref{tab:simulation_setup}.  The  market configurations for comparisons  are shown in TABLE \ref{tab:market_settings}.  We highlight \emph{Optimal P2P} and \emph{Optimal P2P + ES}  as our main focus.   
 

\begin{table}[h]
	\setlength\tabcolsep{3pt}
	\centering
	\caption{Market configurations for comparison}
	\label{tab:market_settings}
	\begin{tabular}{lllll}     
	\toprule
Market &  ES~~~~~~ ~~  & P2P~~~~~ & Network charge~~~  \\
	\hline
No P2P   &   &   & \\
Free P2P  & &  $\checkmark$      &  \\
 Social P2P  &   & $\checkmark$  &  Internalized \\
\textbf{Optimal P2P}   &     & \textbf{$\checkmark$}  & \textbf{$\checkmark$} \\  
No P2P  + ES &  $\checkmark$  &   & \\
Free P2P + ES  & $\checkmark$ &  $\checkmark$      &  \\
 Social P2P + ES  &  $\checkmark$   & $\checkmark$  &  Internalized \\
\textbf{Optimal P2P + ES}~~~  &\textbf{$\checkmark$}      & \textbf{$\checkmark$}  & \textbf{$\checkmark$} \\  
		\bottomrule
	\end{tabular} 
	\end{table}

\begin{figure}[h]
	\centerline{\includegraphics[width=2.8 in]{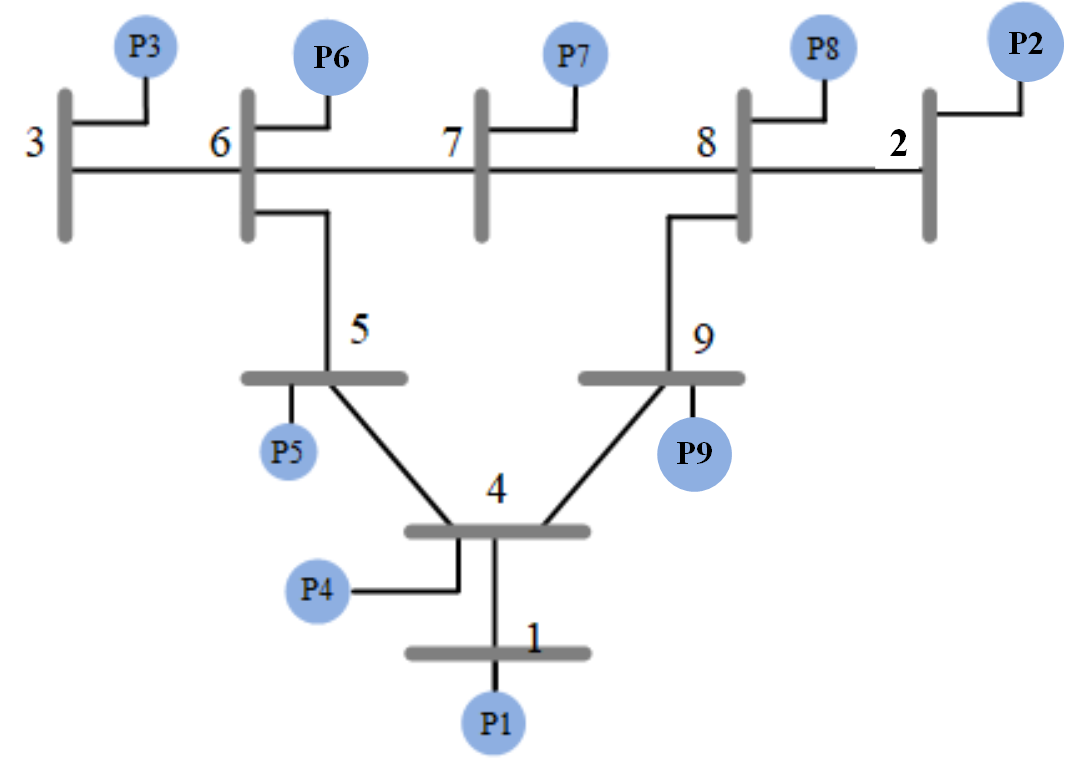}}
		\vspace{-3mm}
	\caption{IEEE-9-bus system with 9 prosumers (P1-P9).}
	\label{fig:IEEE-9-bus}
\end{figure}

\begin{figure}[h]
	\centerline{\includegraphics[width=2.8 in]{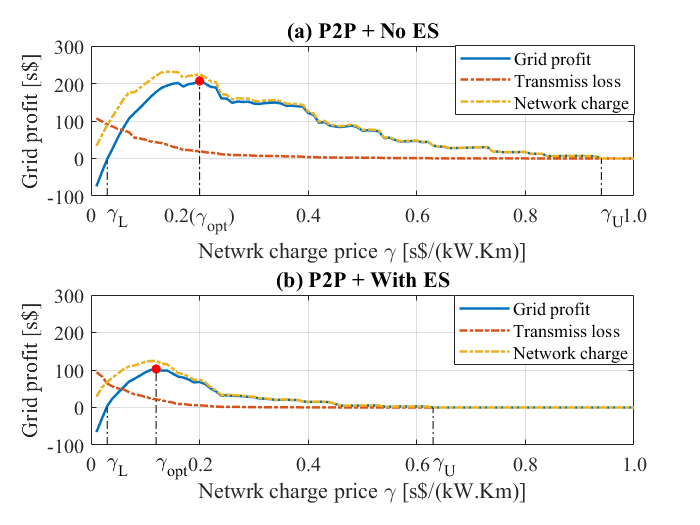}}
	\vspace{-3mm}
	\caption{Grid profit w.r.t. network charge price $\gamma$ for IEEE 9-bus system: (a) P2P + No ES. (b) P2P + With ES. ($\gamma_{\rm L}$: minimum network charge price for the grid to attribute transmission loss. $\gamma_{\rm opt}$: optimal network charge price for maximum grid profit. $\gamma_{\rm U}$: maximum network charge price that the prosumers would take.)}
	\label{fig:Networkcharge_simulation}
\end{figure}
\vspace{-3mm}
\subsection{IEEE 9-bus system}
We first use the small-scale IEEE 9-bus system  with 9 prosumers  shown in Fig. \ref{fig:IEEE-9-bus} to evaluate the proposed optimal network charge model. 
By  solving the \eqref{pp:single-level}, we can obtain the optimal network charge price   $\gamma_{\rm opt} =0.2$ $\si{s\$/(\kilo\watt\cdot \kilo\meter)}$ (No ES) and $\gamma_{\rm opt} =0.12$$\si{s\$/(\kilo\watt\cdot \kilo\meter)}$ (With ES).    To verify the solution accuracy, we compare the obtained solutions with that identified from  simulating the range of network charge price $\gamma \in [0, 1]$$\si{s\$/(\kilo\watt\cdot \kilo\meter)}$  with an incremental of $
 \Delta \gamma = 0.01$ $\si{s\$/(\kilo\watt\cdot \kilo\meter)}$.  For each simulated network charge price, we evaluate the grid profit, network charge, and transmission loss  defined in \eqref{eq:upper-level} and display their changes w.r.t. the network charge price   in Fig. \ref{fig:Networkcharge_simulation}.  From the results,  the optimal network charge price can be identified where the grid profit is maximized, which are  $\gamma_{\rm opt} =0.2$ $\si{s\$/(\kilo\watt\cdot \kilo\meter)}$ (No ES) and $\gamma_{\rm opt} =0.12$$\si{s\$/(\kilo\watt\cdot \kilo\meter)}$ (With ES) corresponding well to the obtained solutions.  This demonstrates the effectiveness of the proposed solution method. 
By further examining the  simulation results, we  can draw the following main conclusions. 
\subsubsection{\textbf{The network charge model admits an equilibrium network charge price}}
From Fig. \ref{fig:Networkcharge_simulation} (a) (No ES) and  \ref{fig:Networkcharge_simulation} (b) (With ES), we observe that the grid profit first approximately increases and begins to drop after reaching the optimal network charge price $\gamma_{\rm opt}$ with the maximum grid profit. Since for any given network charge price $\gamma$, there exists an optimal energy management strategy for the prosumers (i.e., there exists an optimal solution for the lower level problem \eqref{eq:lower-level}),  we imply that $\gamma_{\rm opt}$ is the \emph{equilibrium} network charge price. This demonstrates the existence of  \emph{equilibrium} for the proposed Stackelberg game, which is in line with  \textbf{Theorem} \ref{them:theorem1}. 


Besides, we can imply from the results that there exists a  minimal network charge price for the grid operator to attribute  the transmission loss. Such minimal network charge price occurs where the network charge revenue equals the transmission loss (i.e., the grid operator has \emph{zero} profit). Specifically,  the minimal  network charge price  is  $ \gamma_{\rm L} = 0.03\si{s\$/(\kilo\watt\cdot \kilo\meter)}$ both with and without ES for the tested case.   In addition, we note that there also exists  a  maximal network charge price that the prosumers would  take,   which are  $\gamma_{\rm U} =0.94$$\si{s\$/(\kilo\watt\cdot \kilo\meter)}$ (No ES) and  $\gamma_{\rm U}  = 0.6 \si{s\$/(\kilo\watt\cdot \kilo\meter)}$ (With ES) for the tested case.   When the network charge price exceeds the maximum price, we observe that no transaction happens in the P2P market.   We note that when the prosumers have ES, they  would take lower network charge price.  This is reasonable as the prosumers can  use the ES to shift surplus renewable generation for future use in addition to trade in the P2P market. 
This can be perceived from Fig \ref{fig:Networkcharge_simulation_trades} which shows the total P2P trades  in the P2P market w.r.t. the network charge price  both with and without ES.  We note that less trades will be made when the prosumers have ES than that of  No ES for any specific network charge price.  Moreover,  the total trades drop faster w.r.t. the increase of network charge price when the prosumers have  ES. This implies that the deployment of ES will make the prosumers more sensitive to the network charge price and impact the optimal  network charge price. 

\begin{figure}[h]
	\centerline{\includegraphics[width=3.0 in]{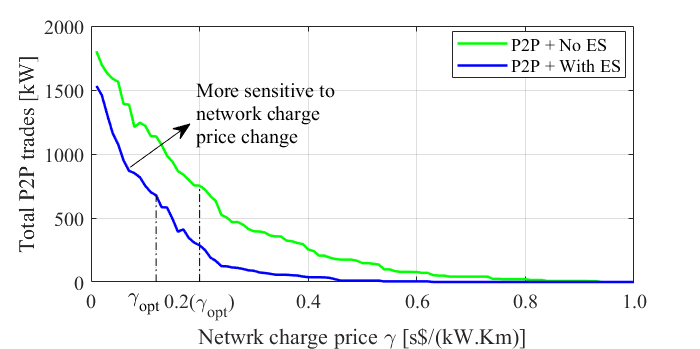}}
	\caption{Total P2P trades w.r.t. network charge price $\gamma$ for IEEE 9-bus system ($\gamma_{\rm opt}$: optimal network charge price).}
	\label{fig:Networkcharge_simulation_trades}
\end{figure}

\subsubsection{\textbf{The network charge  can benefit both the grid operator and the prosumers}}
From the above results, we conclude that the proposed optimal  network charge  can provide positive  profit to  the grid operator. 
This implies that the grid operator can benefit from empowering P2P energy trading. 
An interesting question to ask is how the economic benefit of P2P is shared by the grid operator and the prosumers with the proposed network charge mechanism. 
To answer that question,  we use \emph{No P2P} as the base and define the increased  profit for the grid  operator and the prosumers  as the \emph{benefit} harnessed from a specific P2P market.  
We compare the benefit of the two stakeholders with \emph{Optimal P2P} and \emph{Free P2P}. 
For \emph{Optimal P2P}, we impose the obtained optimal network charge price  $\gamma_{\rm opt}= 0.2$ $\si{s\$/(\kilo\watt\cdot \kilo\meter)}$ (No ES) and $\gamma_{\rm opt } = 0.1$$\si{s\$/(\kilo\watt\cdot \kilo\meter)}$.  For \emph{Free P2P}, we set a sufficiently small network charge price $\gamma=1e-7$$\si{s\$/(\kilo\watt\cdot \kilo\meter)}$  to ensure the uniqueness of solution as mentioned.  We evaluate the benefit of grid operator and  prosumers for each time  period and display the results over the 24 periods   in  Fig. \ref{fig:benefit_distribution_no_ES} (No ES) and Fig. \ref{fig:benefit_distribution_ES} (With ES).  We note that when there is no network charge (i.e., \emph{Free P2P} and \emph{Free P2P + ES}), the prosumers can gain considerable benefit from P2P transaction.  Whereas the grid operator will have to undertake the transmission loss which are in total  116.94s\$ (No ES) and 103.38s\$ (With ES).   Comparatively,   the proposed optimal network charge (i.e., \emph{Optimal P2P}, \emph{Optimal P2P + ES}) can provide positive benefit to both the grid operator and the prosumers, i.e.,  206.77s\$ vs. 158.67s\$ (No ES) and  103.01s\$ vs. 76.50s\$ (With ES).  We therefore imply that the optimal (\emph{equilibrium}) network charge price that maximizes the grid profit also secures the prosumers' profit. 
Moreover, the benefit of P2P is almost equally shared by the grid operator and the prosumers  (i.e., 57.38\% vs. 42.6\%). 

\begin{figure}[h]
	\centerline{\includegraphics[width=3.2 in]{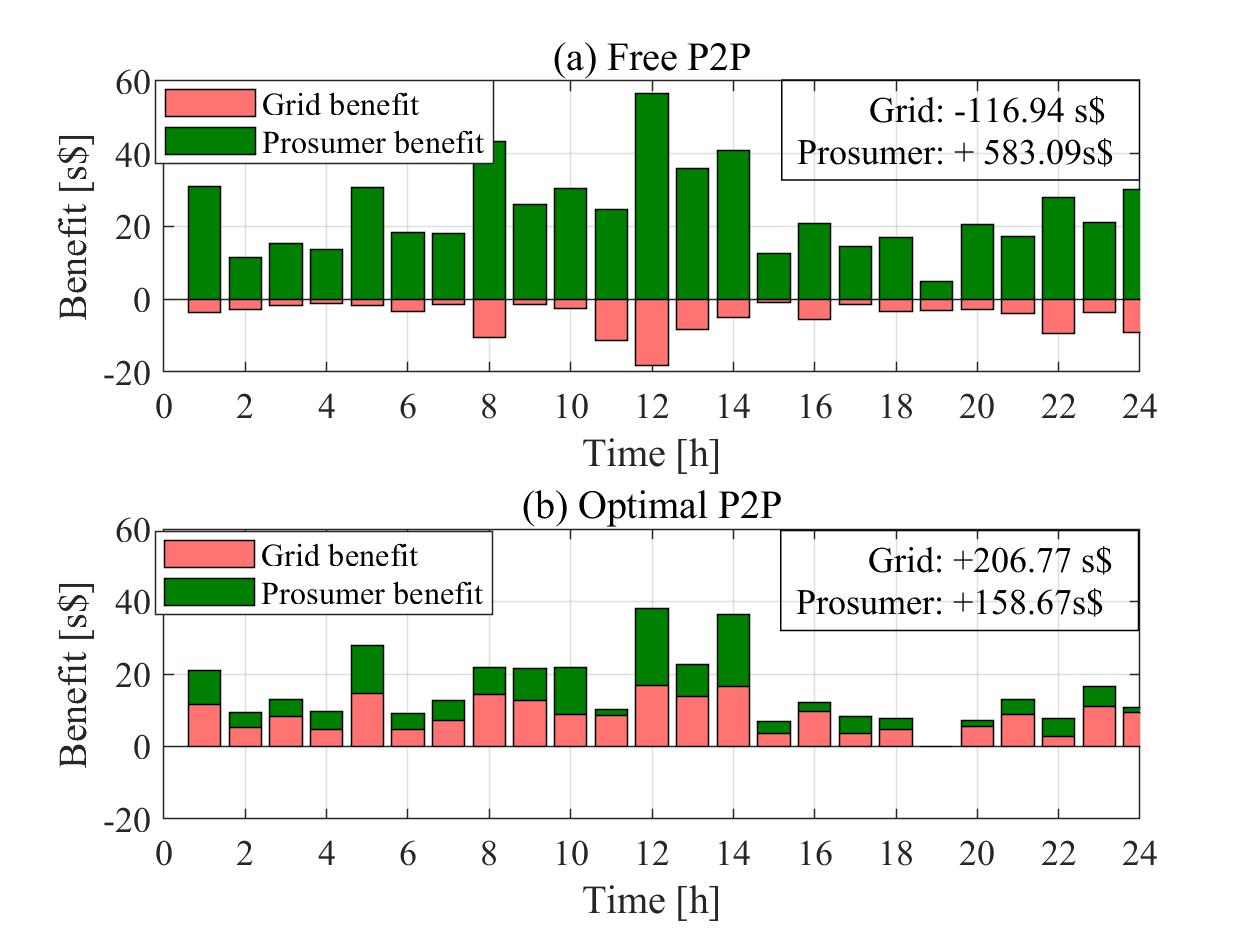}}
	\caption{Benefit of the grid  operator and  prosumers with \emph{Optimal P2P} and \emph{Free P2P} for IEEE 9-bus system (using No P2P as benchmark).}
	\label{fig:benefit_distribution_no_ES}
\end{figure}

\begin{figure}[h]
	\centerline{\includegraphics[width=3.2 in]{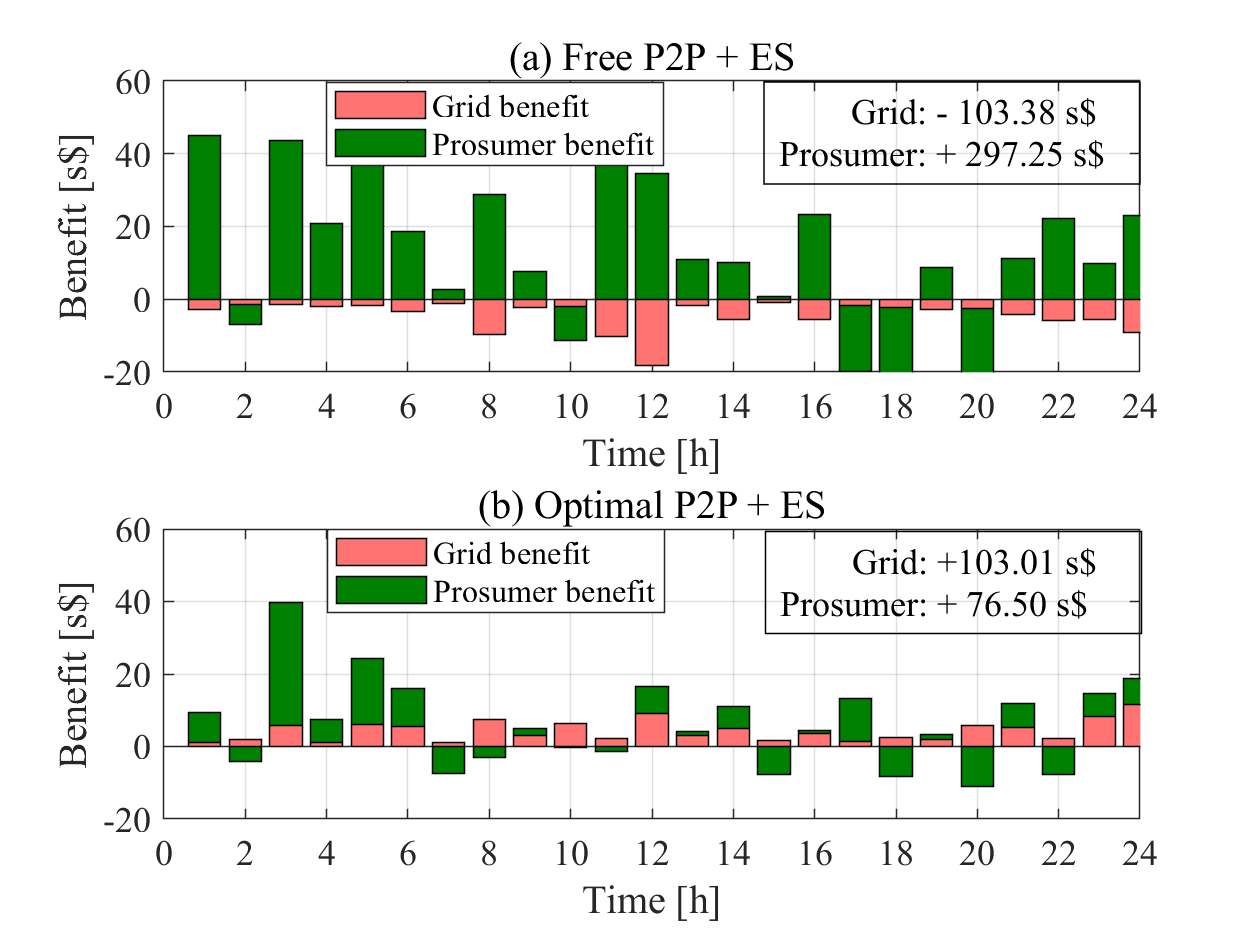}}
	\caption{Benefit of the grid  operator and  prosumers with \emph{Free P2P + ES} and  \emph{Optimal P2P + ES}  for IEEE 9-bus system (using No P2P as benchmark).}
	\label{fig:benefit_distribution_ES}
\end{figure}

\subsubsection{\textbf{The network charge  provides near-optimal social welfare}}
\emph{Social welfare} is one of the most important measures to be considered for market design.  For the concerned P2P market involving the grid operator and the prosumers,  the \emph{social welfare} refers to the sum profit  of  grid operator  and prosumers (i.e., social profit) and defined as   ${\rm Social~profit} = \smallsum_t\smallsum_i U_{i, t}(\PP_{i, t}) - \rho \smallsum_t \smallsum_{(i, j) \in \LL} b_{ij}(\theta_{i, t} - \theta_{j, t})^2$.  In this part, we study the  social profit yield by the proposed network charge model. To identify the social optimality gap,   we compare  \emph{Optimal P2P}  with  \emph{Social P2P}.  
 We evaluate the social profit for each time period with \emph{Optimal P2P} and \emph{Social P2P},  and display the results over the 24 periods in Fig.  \ref{fig:social_profit} (a) (No ES) and Fig.  \ref{fig:social_profit} (b) (With ES). 
To identify the social optimality gap, we fill the difference of social profit curves with \emph{Optimal P2P} and \emph{Social P2P} in blue.  Note that the \emph{positive area} can be interpreted as  the social optimality gap. From the results, we conclude that  the social optimality gap is about 4.70\% (No ES)  and $1.32\%$ (With ES). To be noted, though we observe a larger shaded area for the case with ES (see Fig. \ref{fig:social_profit} (b)),  the accumulated  positive area is quite small. This implies that the proposed network charge mechanism can provide \emph{near-optimal} social welfare.   
 
We further study how the social profit is affected by the network charge price. Similarly, we simulate the range of network charge price  $\gamma \in [0, 1]$$\si{s\$/(\kilo\watt\cdot \kilo\meter)}$  with an incremental of $\Delta \gamma = 0.01$ $\si{s\$/(\kilo\watt\cdot \kilo\meter)}$. For each simulated network charge price, we evaluate the total social profit over the 24 periods.  As shown in Fig. \ref{fig:market_outcome}, we observe that the social profit first increases w.r.t. the network charge and begins to drop after the social optima is reached. Besides, we note that though the obtained optimal network charge price does not coincide with the social optima, the social optimality gap is quite small, which are only $4.70\%$ and $1.32\%$ (With ES) as discussed. 

\begin{figure}[h]
	\centerline{\includegraphics[width=3.2 in]{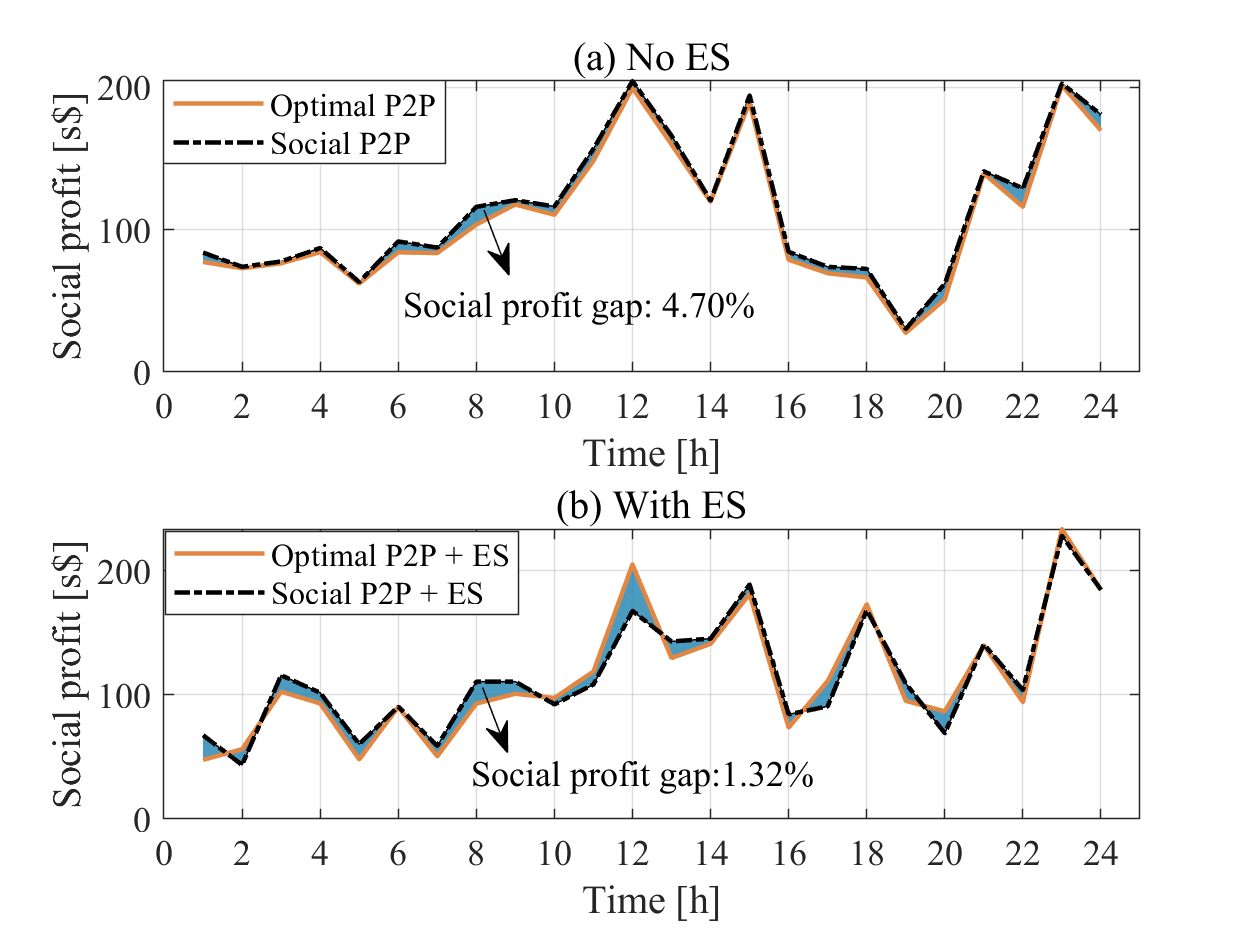}}
	\caption{Social profit for each time period for IEEE-9-bus system: (a) No ES. (b) With ES. (Positive shaded area represents the social profit loss of \emph{Optimal P2P} compared with \emph{Social P2P}).}
	\label{fig:social_profit}
\end{figure}

\begin{figure}[h]
	\centerline{\includegraphics[width=3.2 in]{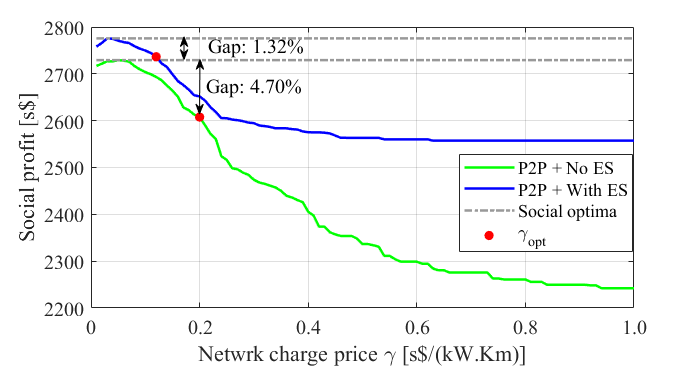}}
	\caption{Social profit w.r.t. network charge price for IEEE-9-bus system. (P2P + With ES: social optimality gap 1.32\%. P2P + No ES: social optimality gap 4.70\%.)}
	\label{fig:market_outcome}
\end{figure}

\subsubsection{\textbf{The network charge favors localized transaction and curbs long distancing transaction}} 
In this part, we study how the proposed network charge  shapes  the P2P  markets. 
We compare  \emph{Optimal P2P} with \emph{Free P2P} both with and without ES on the prosumer side. 
For each market, we calculate the  aggregated  transaction (in \si{\kilo\watt}) for the trading peers over the 24 periods and visualize the transaction in Fig. \ref{fig:trade_both}. The circles with IDs  indicate the prosumers and the line  thickness  represents  the  amounts of  P2P transaction.   We observe that the network charge has an obvious impact on the behaviors of prosumers in the P2P markets.  Specifically,  by comparing Fig. \ref{fig:trade_both} (a)  (No ES) and Fig. \ref{fig:trade_both} (b) (With ES), we notice that the network charge 
dose not affect the transaction between the prosumers in close proximity (e.g., 3-6, 7-8, 2-8, 1-4) but  obviously discourages  the  long distancing  transaction (e.g., 4-6, 1-6, 1-9). 
This is reasonable as  the network charge  counts on the electrical distance.  Therefore,  the proposed  network charge  model favors localized  transaction and curbs the long  distancing transaction. This makes sense considering the  transmission losses related to the long distancing transaction.  For the case with ES, we can draw the similar conclusion  from Fig. \ref{fig:trade_both} (a) (No ES) and  Fig. \ref{fig:trade_both} (b) (With ES). 

\begin{figure}[h]
	\centerline{\includegraphics[width=3.2 in]{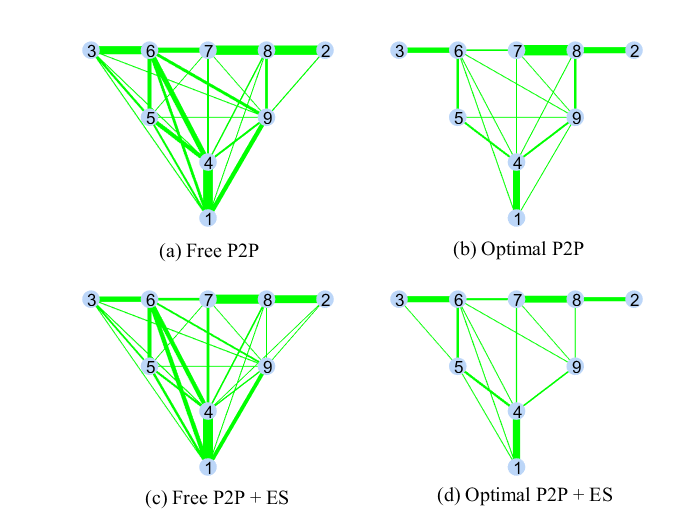}}
	\caption{Total P2P trades of 24 periods across the prosumers for IEEE 9-bus system (line thickness represents the amounts of transaction).}
	\label{fig:trade_both}
\end{figure}

\subsection{IEEE 39-bus, 57-bus and  118-bus systems}

\begin{table*}[h]
	\setlength\tabcolsep{3pt}
	\centering
	\caption{Outcomes of different P2P markets}
	\label{tab:larger-scale-cases}
	\begin{tabular}{cl rrr rr r}     
		\toprule 	
		\multirow{3}{*}{\textbf{System}} &
		\multirow{3}{*}{\textbf{Market}}  & 
		\multicolumn{3}{c}{\textbf{Grid-wise}} &  
		\multicolumn{2}{c}{\textbf{Prosumer-wise}} &  
		\multicolumn{1}{c}{\textbf{System-wise} }  \\ 
		
		\cmidrule(lr){3-5}
		\cmidrule(lr){6-7}
		\cmidrule(lr){8-8}
		
		& &  Transmission loss
		&  Network charge
		& Grid profit    
		&  Prosumer  profit 
		&  Total transaction 
		& Social profit  \\
		
		 &  & \multicolumn{1}{c}{$\times10^2$[s\$] } 
		&  \multicolumn{1}{c}{$\times10^2$[s\$]} 
		&  \multicolumn{1}{c}{$\times10^2$[s\$]}
		&  \multicolumn{1}{c}{ $\times10^2$[s\$]} 
		& \multicolumn{1}{c}{ $\times10^2$[\si{\kilo\watt\hour}] }
	   &  \multicolumn{1}{c}{$\times10^2$[s\$]}\\   
		\midrule  
		\multirow{8}{*}{9-bus}  
& No P2P               &  0          & 0               & 0                 &  22.42      &  0            &  22.42             \\
&Free P2P             & 1.17       &  0              & -1.17          &   28.25      &  17.43      & 27.10      \\ 
&Social P2P           & --           & --              & --                 &  --             & --               & 27.35 		 \\   
& \textbf{Optimal P2P}     &  \textbf{0.19}        &  \textbf{2.26}        &  \textbf{2.07}            & \textbf{24.00}         &  \textbf{7.82}        & \textbf{26.08}         \\  
		 \cline{2-8}
& No P2P+ES              & 0            &  0               &  0                 &   25.57        &   0             &   25.57       \\
&Free P2P+ES          &  1.03      &  0                & -1.03            &   28.54       &    15.57      &  27.51     \\ 
&Social P2P+ES         &--              & --                & --                 &  --               &   --                 &  27.87	        		 \\   
& \textbf{Optimal P2P+ES}     &  \textbf{0.28}        &  \textbf{1.22}        &  \textbf{0.94}              &  \textbf{26.56}         &  \textbf{7.64}        & \textbf{27.50}           \\  
\midrule 
\multirow{8}{*}{39-bus}  
& No P2P              &0          &  0              &0               &  110.86          &    0                &   110.86  \\
& Free P2P           &3.27      &   0             &  -3.27       & 151.03       &  112.64           &   147.76        \\ 
& Social P2P        & --          & --               & --              &  --                &   --                  &     148.15    \\     
& \textbf{Optimal P2P}    & \textbf{0.55}     & \textbf{14.79}        &  \textbf{14.24}               &   \textbf{123.44}         
&  \textbf{52.33}           &      \textbf{137.68}           \\  		
\cline{2-8}
& No P2P+ES             & 0         &  0               & 0                &  124.69      &    0                &    124.69  \\
& Free P2P + ES         & 2.62      &  0              & -2.62          & 154.38       &  107.69           &    151.76  \\ 
& Social P2P + ES       & --            & --              & --                &  --             &   --                  &    152.18     \\       	
& \textbf{Optimal P2P+ES}    & \textbf{0.60}     & \textbf{11.26}        & \textbf{10.65}              &  \textbf{134.07}      
&  \textbf{ 45.17}      &     \textbf{144.73}          \\
\midrule   
\multirow{8}{*}{57-bus}  
& No P2P            & 0             & 0                    & 0                 &  191.63     & --           &  191.63   \\
& Free P2P         &  65.52       &  0                 &  -65.52          &  26.20     &   185.61    & 196.53  \\ 
& Social P2P       & --             & --                   & --                  & --           & --             &  232.05\\  
& \textbf{Optimal P2P}   &  \textbf{4.72}        & \textbf{22.40}        &  \textbf{17.68}      &          \textbf{205.94}      &   \textbf{61.16}          &  \textbf{223.61}     \\
		 \cline{2-8}
& No P2P + ES             & 0          & 0                 &0              &   212.50    & 0          &  212.50         \\
& Free P2P + ES          &  66.20   & 0                 &   -66.20           &  272.92   &  177.23     &206.75  \\ 
& Social P2P + ES       &--                  & --                   & --                         &  --     &   --                    &   240.19\\        
& \textbf{Optimal P2P+ES}   &  \textbf{5.39}        & \textbf{18.49}        &  \textbf{13.10}      &          \textbf{222.52}      &   \textbf{52.45}          &  \textbf{235.62}     \\    
		\midrule   
\multirow{8}{*}{118-bus}  
& No P2P               &   0        &  0             & 0          &427.68     & 0        & 427.68    \\
& Free P2P             & 111.20       &  0               & -111.20             &  567.38  &   367.09       &     455.64 \\ 
& Social P2P           & --                    & --                   & --                         &  --     &   --                    &    530.84\\ 
&\textbf{Optimal P2P}        &   \textbf{5.33}         & \textbf{46.97}        &  \textbf{41.63}             &  \textbf{46.76}       &\textbf{149.59}          &  \textbf{509.18}    \\
\cline{2-8}
& No P2P  + ES             &   0        & 0               &0            & 474.84   & 0       &  474.84    \\
& Free P2P  + ES           & 177.47     &  0           &  -177.47            &  603.83      &  391.05      &   426.35  \\ 
& Social P2P  + ES         & --                 & --                   & --                         & --   &   --                    &   557.79 \\    
&\textbf{Optimal P2P + ES}        &   \textbf{5.46}         & \textbf{37.43}        &  \textbf{31.97}             &  \textbf{505.38}       &\textbf{128.27}          &  \textbf{537.35}    \\     
		\bottomrule 
	\end{tabular}
\end{table*}

We further examine the performance of the proposed network charge mechanism  by  simulating the  IEEE 39-bus,  57-bus, and 118-bus systems.  We follow the  same  simulation set-ups in Section IV-A and  compare  the different markets  in TABLE \ref{tab:market_settings}.
We report  the results  for  different markets  and bus systems  in TABLE \ref{tab:larger-scale-cases}.  
Particularly, we  group the results by  \emph{Grid-wise},  \emph{Prosumer-wise} and \emph{System-wise}. 
For \emph{Grid-wise}, we study the total transmission loss, network charge revenue and the grid profit.  For \emph{Prosumer-wise},  we are concerned  with the total  prosumer profit and total P2P transaction. 
For \emph{System-wise},  we  evaluate the social profit (i.e., grid operator profit plus prosumers' profit). 
Note that the results for IEEE 9-bus system are also included for completeness. The results associated  with \emph{Optimal P2P}  and \emph{Optimal P2P + ES} have been highlighted in bold as our main focus. 
Overall, for the larger electrical networks, we can draw similar results in Section IV-B.   

Specially,   the proposed optimal network charge can provide positive  profit  both to the power grid and the prosumers as with \emph{Optimal P2P} and \emph{Optimal P2P + ES} reported  in TABLE \ref{tab:larger-scale-cases}. 
Whereas \emph{Free P2P} and \emph{Free P2P + ES}  only favor the prosumers with considerable profit increase over \emph{No P2P} and will displease the power grid operator considering the uncovered transmission loss (i.e., negative profit for the grid operator). 
This implies that the network charge  is necessary to  enable  the successful deployment of P2P  market in the existing power system from the perspective of economic benefit. 

Besides, we can conclude that the proposed network charge is favorable considering the benefit of P2P shared by the grid operator and the prosumers.  Similarly, using \emph{No ES} as the benchmark, we define the  profit increase of the grid operator and the prosumers  as the benefit.  Based on the results in TABLE \ref{tab:larger-scale-cases}, we have the report regarding the benefit of grid operator and the prosumers 
 in  Fig. \ref{fig:benefit_distribution_large} (a)  (No ES) and \ref{fig:benefit_distribution_large} (b) (With ES).   
Notably, we see that the grid operator and the prosumers achieve almost equal benefit from the P2P market with all cases.  Specifically,  the benefit for the prosumers and  grid operator are about 49.0\% vs. 51.0\% (No ES)  and  48.4\% vs. 51.6\% (With ES) for the tested IEEE-118 bus systems.   This makes sense considering  the balance of the P2P market. 
\begin{figure}[h]
	\centerline{\includegraphics[width=3.4 in]{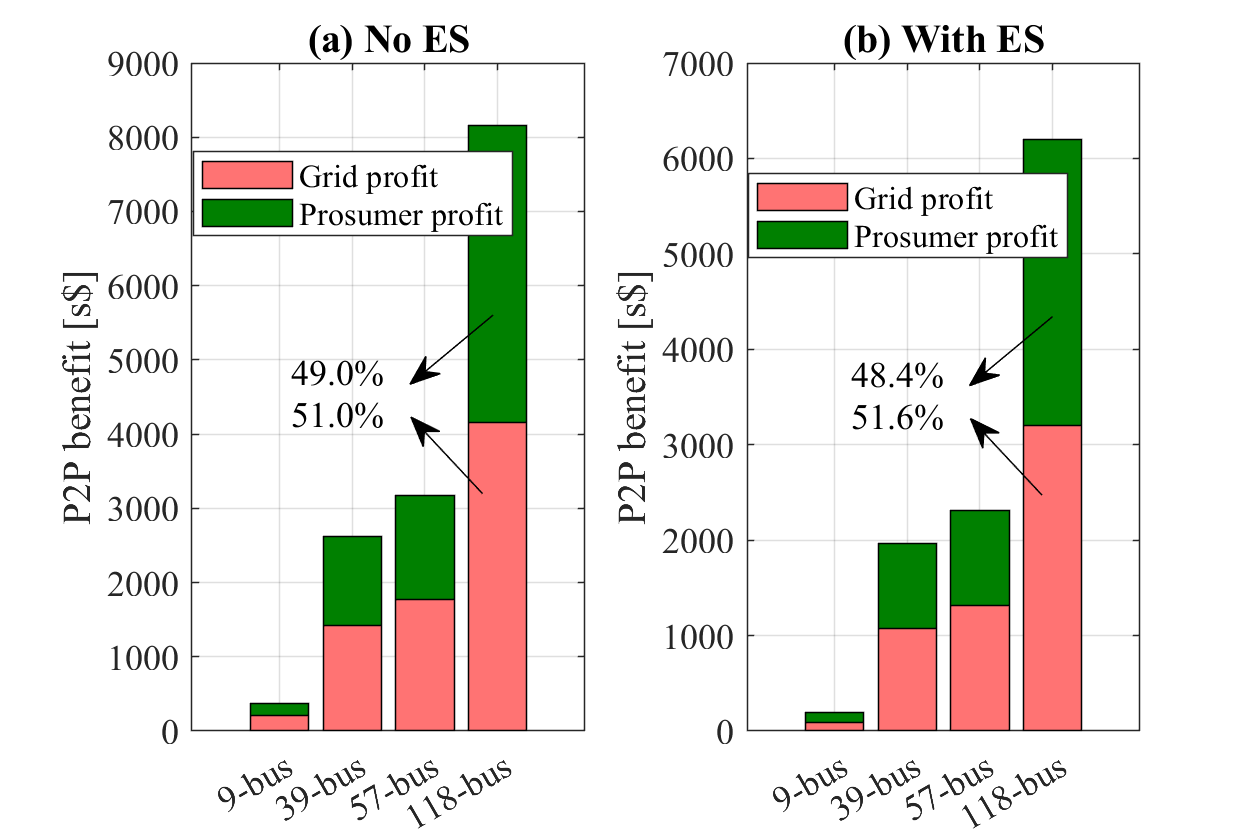}}
	\vspace{-3mm}
	\caption{Benefit of grid operator and the prosumers with \emph{Optimal P2P}: (a) No ES. (b) With ES. (using \emph{No P2P} as benchmark)}
	\label{fig:benefit_distribution_large}
\end{figure}

In addition,  we  conclude that the  network charge  can be used to shape the P2P markets.  For the tested bus systems, we compare the P2P  transaction of  \emph{Free P2P} and \emph{Optimal P2P} both with and without ES. 
Similarly, we visualize the total transaction over the 24 periods for the trading peers in  Fig.
\ref{fig:39-bus-transaction} (39-bus), \ref{fig:57-bus-transaction} (57-bus), \ref{fig:118-bus-transaction} (118-bus).  The circles with IDs indicate prosumers located at the buses  and  line thickness  represents the amounts of  transactions.   We note that the imposed network charge has an obvious impact on the energy trading behaviors of prosumers in the P2P market.  When there is no network charge and the grid operator has no manipulation on the P2P market,  the prosumers could  trade  regardless of the electrical distances, leading to massive long distancing trades. The could be problematic considering the high transmission loss and the possible network violations taken by the grid operator. 
Comparatively, the proposed network charge  favors localized transaction and discourages long distancing transaction,  yielding much lower transmission loss as reported in TABLE \ref{tab:larger-scale-cases}  (see Column 3). More importantly,  the  network charge is necessary for the grid operator to ensure the network constraints. 

 Last but not the least,  we conclude that the proposed network charge favors  social welfare. 
By examining the  \emph{System-wise} performance indicated  by  social profit in TABLE \ref{tab:larger-scale-cases} (Column 8),   we  notice that P2P transaction (i.e., \emph{Optimal P2P},  \emph{Free P2P} and  \emph{Social P2P})  favors  social profit over \emph{No P2P}.   More notably,  we find that \emph{Optimal P2P} and \emph{Optimal P2P + ES} provide  near-optimal social profit indicated by  \emph{Social P2P}.   Specifically,  the overall social optimality gap  is less than 7\% (No ES) and less than 5\%  (With ES) with \emph{Optimal P2P}  as  reported in Fig. \ref{fig:social_optimality},

 \begin{figure}[h]
 	\centerline{\includegraphics[width=3.4 in]{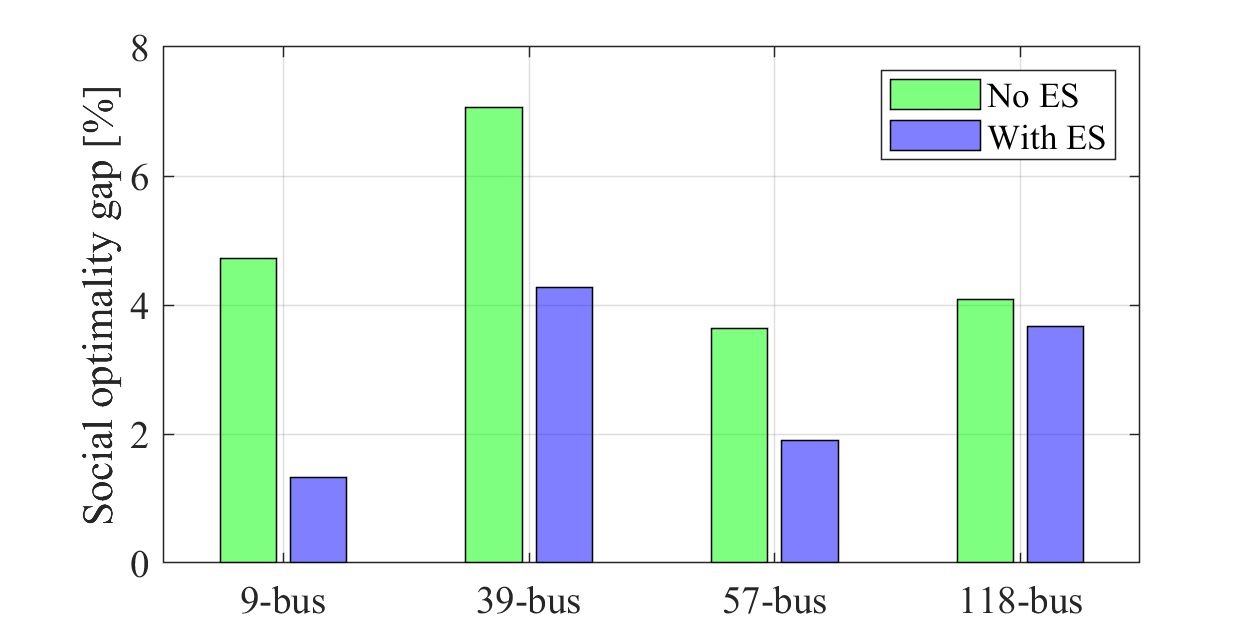}}
 	\vspace{-3mm}
 	\caption{Social optimality gap of \emph{Optimal P2P}.}
 	\label{fig:social_optimality}
 \end{figure}

\begin{figure}[h]
	\centerline{\includegraphics[width=3.2 in]{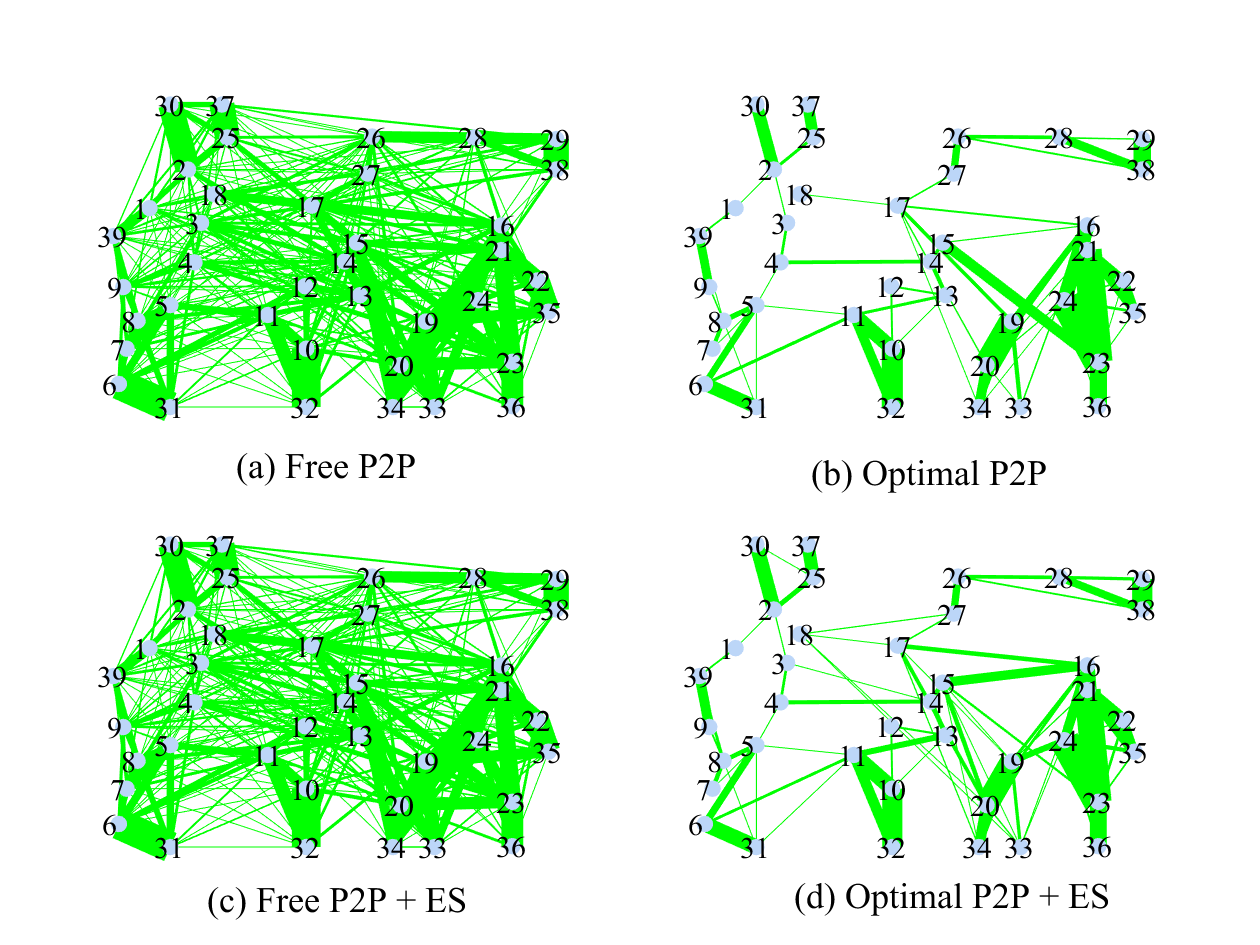}}
	\caption{Total P2P trades of 24 periods across the prosumers for IEEE 39-bus system (line thickness represents the amount of transaction).}
	\label{fig:39-bus-transaction}
\end{figure}
\begin{figure}[h]
	\centerline{\includegraphics[width=3.4 in]{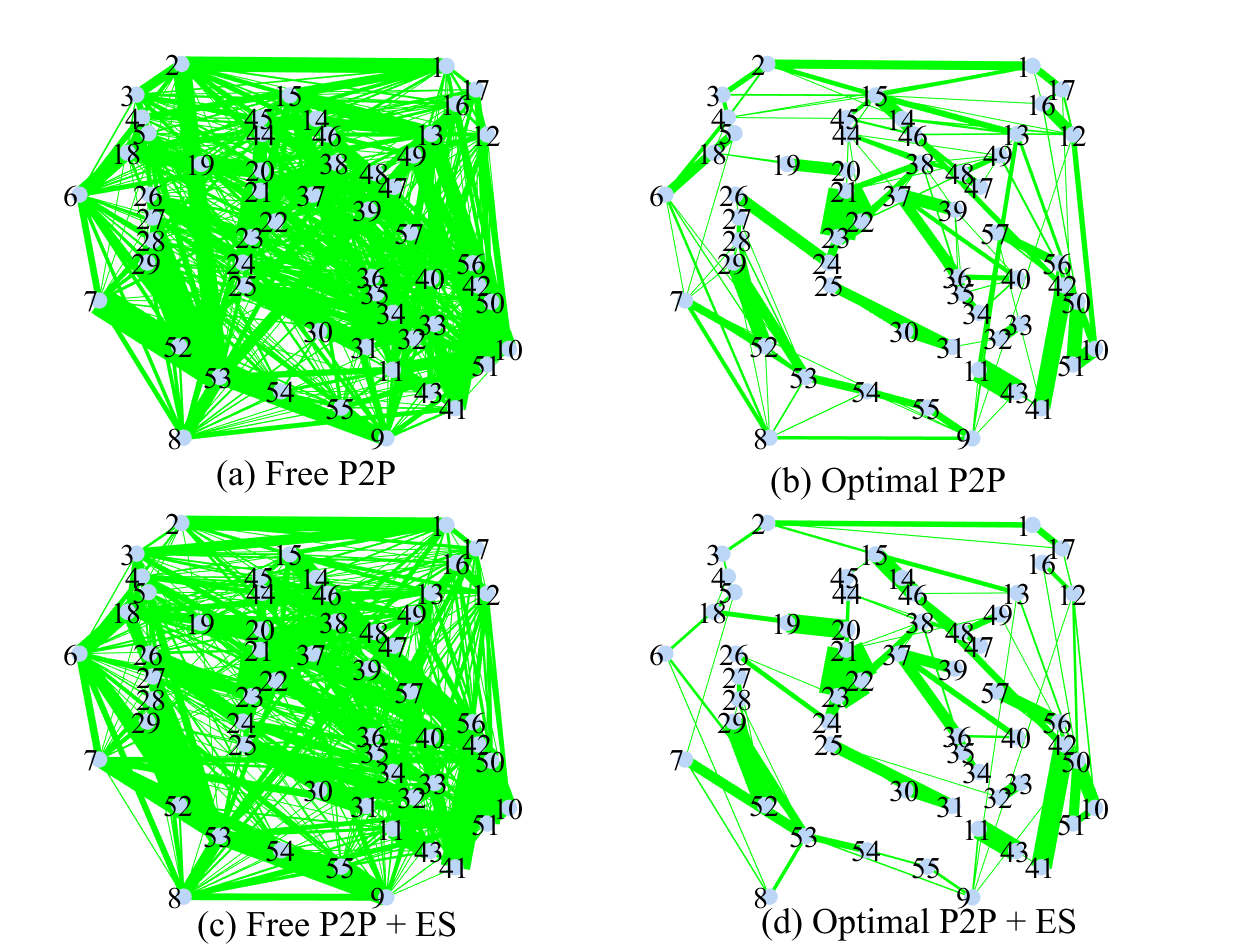}}
	\caption{Total P2P trades of 24 periods across the prosumers for IEEE 57-bus system (line thickness represents the amount of transaction).}
	\label{fig:57-bus-transaction}
\end{figure}
\begin{figure}[h]
	\centerline{\includegraphics[width=3.4 in]{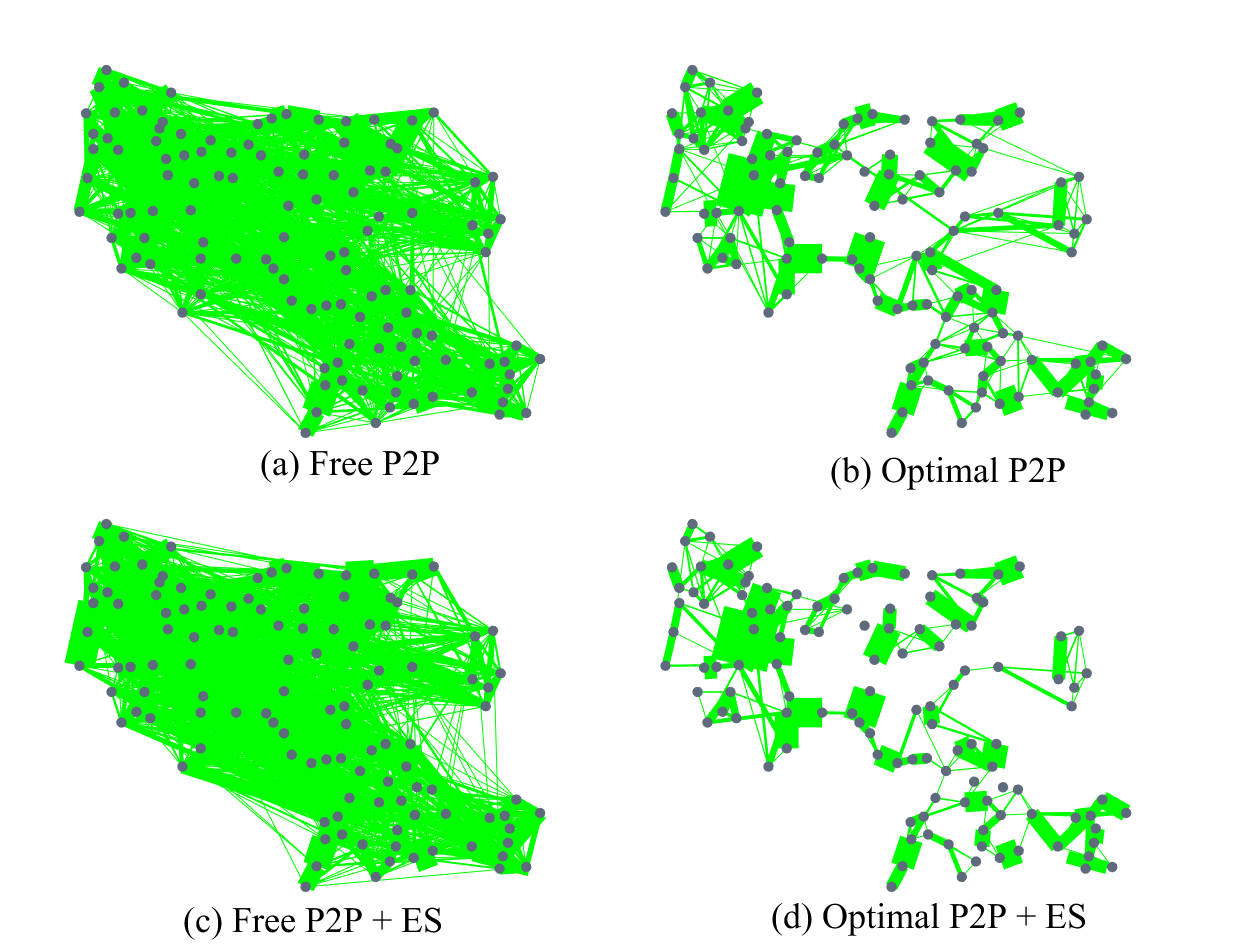}}
	\caption{Total P2P trades of 24 periods across the prosumers for IEEE 118-bus system (line thickness represents the amount of transaction).}
	\label{fig:118-bus-transaction}
\end{figure}

\section{Conclusion and Future Works}
This paper discussed the integration of the P2P  market scheme into the existing power systems from the perspective of network charge design. We used network charge as a means for the grid operator to attribute  grid-related cost (i.e., transmission loss) and ensure network constraints for empowering P2P transaction.
We characterized the interaction  between the power grid operator and the prosumers in a P2P market as a Stackelberg game. The grid operator first decides on the optimal network charge price to trade off the network charge revenue and the transmission loss considering the network constraints,  and then the prosumers optimize their energy management (i.e., energy consuming, storing and trading) for maximum economic benefit. We proved the Stackelberg game admits an \emph{equilibrium} network charge price. 
Besides, we proposed a solution method to obtain the \emph{equilibrium} network charge price  by converting the bi-level optimization problem into a single-level optimization problem. By simulating the IEEE bus systems, we demonstrated that the proposed network charge mechanism can benefit both the grid operator and the prosumers and achieve \emph{near-optimal} social welfare. In addition, we found  that the presence of ES on prosumer side will make the prosumers more sensitive to the network charge price increase. 

In this paper, we have studied the optimal network charge with deterministic supply and demand  and found that the network charge is effective in shaping the behaviors of prosumers in a P2P market.  Some future works  along this line include:  1) designing optimal network charge  price considering the uncertainties of prosumer supply and demand;  2) using the network charge as a tool to achieve demand response.

\bibliographystyle{ieeetr}
\bibliography{reference}

\end{document}